\documentclass[10pt,journal,compsoc]{IEEEtran}
\usepackage{amsmath,amsfonts}
\usepackage{algorithmic}
\usepackage{array}
\usepackage[caption=false,font=normalsize,labelfont=sf,textfont=sf]{subfig}
\usepackage{textcomp}
\usepackage{stfloats}
\usepackage{url}
\usepackage{verbatim}
\usepackage{graphicx}
\usepackage{bm}
\usepackage{breqn}
\def\BibTeX{{\rm B\kern-.05em{\sc i\kern-.025em b}\kern-.08em
    T\kern-.1667em\lower.7ex\hbox{E}\kern-.125emX}}
\usepackage{balance}
\usepackage{breqn}
\usepackage{makecell}
\usepackage{epsfig, epsf, graphicx, float, color}
\usepackage{pstricks, psfrag}
\usepackage{amssymb, amsmath, amsthm, amsfonts, exscale}
\usepackage{verbatim,enumerate,hyperref}
\usepackage{setspace,mathtools}
\usepackage[numbers,sort&compress]{natbib}
\usepackage{steinmetz}
\usepackage{tikz-qtree,tikz-qtree-compat}
\usepackage{tikz,pgf}
\usepackage{graphbox}
\usepackage{enumitem}
\usetikzlibrary{positioning,patterns}
\usetikzlibrary{calc,fadings,decorations.pathreplacing,arrows,
datavisualization.formats.functions,shapes.geometric}
\usepackage[labelfont=bf]{caption}
\usepackage[utf8]{inputenc}
\usepackage[english]{babel}
\usepackage{amsmath}
\usepackage{smartdiagram}
\usepackage{wrapfig,floatflt}
\usepackage{tikz}
\usepackage{tikz-network}
\usetikzlibrary{arrows, automata, positioning, quotes}
\usepackage{lipsum}
\usepackage{caption, subcaption}
\usepackage[numbers,sort&compress]{natbib}
\usepackage[displaymath, mathlines]{lineno}

\usepackage{amsmath,amsfonts}
\usepackage{algorithmic}
\usepackage{algorithm}
\usepackage{array}
\usepackage[caption=false,font=normalsize,labelfont=sf,textfont=sf]{subfig}
\usepackage{textcomp}
\usepackage{stfloats}
\usepackage{url}
\usepackage{verbatim}
\newtheorem{theorem}{Theorem}

\newtheorem{proposition}{Proposition}
\newtheorem{remark}{Remark}

\newtheorem{definition}{Definition}
\allowdisplaybreaks

\usepackage{tikz,soul}
\tikzset{
  treenode/.style = {shape=rectangle, rounded corners,
                     draw, align=center,
                     top color=white, bottom color=blue!20},
  root/.style     = {treenode, font=\Large, bottom color=red!30},
  env/.style      = {treenode, font=\ttfamily\normalsize},
  dummy/.style    = {circle,draw}
}

\usetikzlibrary{calc,fadings,decorations.pathreplacing,arrows,
  datavisualization.formats.functions,shapes.geometric}
\usetikzlibrary{hobby}
\usetikzlibrary{positioning,calc}
\usetikzlibrary{arrows,decorations.markings}
\usesmartdiagramlibrary{additions}

\newcommand{\corb}[1]{\textcolor{black}{#1}}

\newcommand{\lcol}[1]{\textbf{\textcolor{black}{#1}}}

\newcommand{\bbR}{\mathbb{R}}

\newcommand{\bbN}{\mathbb{N}}

\newcommand{\bbP}{\mathbb{P}}

\newcommand{\bbC}{\mathbb{C}}

\newcommand{\bx}{\mathbf{x}}

\newcommand{\bv}{\mathbf{v}}
\newcommand{\bz}{\mathbf{z}}
\newcommand{\bs}{\mathbf{s}}

\newcommand{\blambda}{\boldsymbol{\lambda}}

\newcommand{\balpha}{\boldsymbol{\alpha}}

\newcommand{\E}{\mbox{E}}

\newcommand{\esets}[1]{{\mathbb E} [ #1 ] }

\newcommand{\mcA}{{\mathcal A}}
\newcommand{\mcB}{{\mathcal B}}
\newcommand{\mcC}{{\mathcal C}}
\newcommand{\mcD}{{\mathcal D}}

\newcommand{\mcF}{\mathcal{F}}

\newcommand{\mcN}{{\mathcal N}}

\newcommand{\pset}{{\mathbb P}}

\newcommand{\eset}[1]{{\mathbb E} \left[ #1 \right] }

\newcommand{\id}{\,\mbox{d}}

\def\R{\Bbb{R}}

\theoremstyle{remark}

\newcommand{\argmin}{\operatornamewithlimits{argmin}}
\newcommand{\arginf}{\operatornamewithlimits{arginf}}

\definecolor{darkgreen}{RGB}{25,51,0}
\definecolor{bloodred}{RGB}{101,28,50}

\tikzstyle{block} = [draw,rounded corners=1ex,fill=blue!20,minimum width=2em]
\tikzstyle{emptyblock} = [draw,minimum width=2em]
 
\tikzstyle{branch}=[fill,shape=circle,minimum size=3pt,inner sep=0pt]
\tikzstyle{vecArrow} = [thick, blue,
decoration={markings,mark=at position 1 with
   {\arrow[semithick,blue]{open triangle 60}}}, double distance=1.4pt,
   shorten >= 5.5pt, preaction = {decorate}, postaction = {draw,line
   width=1.4pt, white,shorten >= 4.5pt}]

\begin{document}

\title{Uncertainty quantification of receptor ligand binding sites prediction}

\author{Nanjie Chen, Dongliang Yu, Dmitri Beglov, Mark Kon, Julio Enrique Castrill\'on-Cand\'as 
\thanks{N.Chen is with Graduate School
of Mathematics and Statistics, Boston University, Boston, USA}
\thanks{D.Yu is with Graduate School
of Mathematics, Stony Brook University, Stony Brook, USA}
\thanks{D.Beglov is with the Faculty of Biomedical Engineering, Boston University, Boston, USA.}
\thanks{M.Kon is with the Faculty
 of Mathematics and Statistics, Boston University, Boston, USA}

\thanks{J.Castrill\'on-Cand\'as  is with the Faculty
of Mathematics and Statistics, Boston University, Boston, USA}
}
\IEEEtitleabstractindextext{%
\begin{abstract}
Recent advances in protein docking site prediction have highlighted the limitations of traditional rigid docking algorithms, such as PIPER, which often neglect critical stochastic elements such as solvent-induced fluctuations. These oversights can lead to inaccuracies in identifying optimal docking sites. To address this issue, this work introduces a novel model in which molecular shapes of ligand and receptor are represented using multivariate Karhunen-Loeve (KL) expansions. This method effectively captures the stochastic nature of energy manifolds, allowing for more accurate representations of molecular interactions. Developed as a plug-in for PIPER, our scientific computing software enhances the platform, delivering uncertainty measures for the energy manifolds of ranked binding sites. Our results demonstrate that top-ranked binding sites, characterized by lower uncertainty in the stochastic energy manifold, align closely with actual docking sites. Conversely, sites with higher uncertainty correlate with less optimal docking positions. This distinction not only validates our approach but also sets a new standard in protein docking predictions, offering valuable implications for future molecular interaction research and drug development.
\end{abstract}

\begin{IEEEkeywords}
Karhunen-Loeve Expansion, Proper Orthogonal Decomposition, Stochastic Modeling, Uncertainty Quantification
\end{IEEEkeywords}
}

\maketitle


\section{Introduction}
\IEEEPARstart{T}{he} exploration of protein docking and broader molecular interactions has become a significant focus within the field of biology. 
Cfiberbundle@yeah.net.omputational approaches have already demonstrated success in identifying potential compounds for the treatment of novel diseases, thereby expediting the drug design process. This is especially relevant and vital considering the recent global health crises triggered by pandemics.\par
A popular computational approach for predicting receptor-ligand binding sites is known as rigid body docking, where the molecular shapes are assumed to be fixed (Piper \cite{PIPER2006}, ZDock \cite{Zdock2014}, FFT , FMFT \cite{Padhorny2016}, etc). 
While this method is computationally efficient, rigid docking can fail to accurately predict the binding site if:
i) Either the shape of the receptor or ligand varies significantly during the binding process.
ii)  The conformational shape of the
receptor or ligand is uncertain due to the kinetic presence of the 
solvent atoms.
To address the latter, the flexible Docking approach was developed \cite{bindingML}.
In contrast to rigid docking, where both receptor and ligand are treated as rigid bodies, flexible docking acknowledges the potential conformational molecular changes that can occur during the binding process. While this approach can yield a more accurate prediction, it simultaneously intensifies computational complexity. And the accuracy of flexible docking is greatly dependent on modeling of the receptor and ligand's 
flexibilities, as well as the initial state or conformation of the molecules \cite{totrov2008flexible}\cite{Huang2010}. Further, the flexible docking method tackles conformational changes by  integrating molecular flexibility in a deterministic manner. \cite{FlexibleDocking}. 

However, to address ii), in considering thermal random fluctuations that originate from the solvent, it becomes more practical to incorporate conformational uncertainty as random fields. This necessitates the development of a docking method that more realistically models receptor-ligand binding sites under conformational uncertainty, balancing both accuracy and efficiency. In this paper, we introduce a stochastic framework for assessing the inherent  uncertainty in rigid docking. Our future work will strive to extend this methodology to encompass other docking methods.

Computational receptor-ligand interactions (Docking calculations) involve two methodological selections. The first involves a  goodness of fit, sometimes called a scoring function, which assigns a numerical quality measure to each configuration of the two bodies. This in turn defines an energy manifold parametrized by the molecular spatial degrees of freedom, on which the optimal binding site is sought. The second selection involves the choice of search algorithm on the energy manifold. Both of these choices are based on some assumed molecular model. The partially heuristic scoring function incorporates various aspects of molecular properties, including electron density representations of the molecular shape, electrostatic (see \lcol{Figure} \ref{figure:trypsin} for an example
of potential fields for the Trypsin protein) and solvation terms \cite{Gabb1997,ZDOCK2003}, and structure-based interaction potentials \cite{PIPER2006, Mintseris2007}.

One popular approach to rigid docking is based on the Fast Fourier Transform (FFT). The mathematical formulation of this method is as follows.
Let $\mathbf{x} := [x,y,z]^T \in \mathbb{R}^3$ be spatial coordinates, $ \bm{\alpha} := [\alpha ,\beta, \lambda]^T \in \mathbb{R}^3$ be the rotational coordinates$, \text{and } \bm{\lambda} := [\lambda,\mu,\nu]^T \in \mathbb{R}^3$ be the translational coordinates. For $p = 1,...,P$, let $L_p(\bx) : \mathbb{R}^3 \to \mathbb{C}$, $R_p(\bx) : \mathbb{R}^3 \to \mathbb{C}$ be respectively the different ligand and receptor molecular property maps. The scoring energy function of the receptor-ligand interaction is given by:
$$
\mathrm{E}(\boldsymbol{\alpha}, \boldsymbol{\lambda})
:=\sum_{k=1}^{P} \int_{\mathbb{R}^{3}} \overline{R_{k}(\bx)}\left(T(\boldsymbol{\lambda}) D(\boldsymbol{\alpha}) L_{k}(\bx)\right) \mathrm{d} \bx
$$
where $T:\mathbb{R}^3 \to \mathbb{R}^3$ is the translation operator,  $D:\mathbb{R}^3 \to \mathbb{R}^3$ is the rotation operator. The goal is to attain the optimal docking site by searching the six dimensional function $\mathrm{E}(\boldsymbol{\alpha}, \boldsymbol{\lambda})$ for the minimal energy, that is, finding the docking site $(\boldsymbol{\alpha}_{bind}, \boldsymbol{\lambda}_{bind})$ such that  $\mathrm{E}(\boldsymbol{\alpha}, \boldsymbol{\lambda})$ is minimized. 
\par The key weakness of rigid docking is that the geometry of the molecule
is assumed to be deterministic and fixed. Indeed only the crystalline
structure in many cases is known (Protein Data Bank
\cite{Berman2000}). 
However, uncertainty in the molecular shape can
lead to a predicted erroneous docking site.  The {\it true}
molecular conformation can 
in fact significantly involve the incorporation of different docking
sites, in a probabilistic mixture (ensemble). In particular, thermal fluctuations and solvent interactions,
among other factors, lead to varying conformations of the protein.
In \cite{Allen2004, Neumaier} 
molecular dynamics are used to describe particle movements, using
stochastic initial velocities. Nonetheless, this
model assumes that the particles are in a vacuum, while
interactions with the solvent are ignored. In contrast, 
NAMD is a popular molecular dynamic software for simulating molecular dynamics in a solvent \cite{NAMD}.  Other approaches to molecular dynamics are
based on Langevin
dynamics \cite{Ceriotti2009,Hanggi1995,Jung1987} and Markov random models \cite{Xia2013}.

Due to the random fluctuations of a protein
in its solvent,  the molecular shape becomes stochastic.
Let $\Omega$ denote the set of all possible outcomes in a
complete probability 
space $(\Omega,\mcF,\bbP)$, with $\mcF$ a $\sigma$-algebra of events
and $\bbP$ the probability measure.  Given dynamic thermal randomness
and the resulting conformational 
uncertainties, the receptor and ligand molecular maps will now depend
on a random parameter $\omega \in \Omega$ i.e. for $k = 1, \dots P$
we have $R_k(\bx,\omega):
\R^3 \times \Omega
\rightarrow \bbC
$ and $L_k(\bx,\omega):\R^3 \times \Omega
\rightarrow \bbC$ and the energy function can be expressed as $$\E(\balpha,\blambda,\omega):= \sum_{k=1}^{P}
\int_{\R^{3}} \overline{R_k(\bx,\omega)} T(\blambda)
  D(\balpha)L_{k}(\bx,\omega) \,\mbox{d}\bx.$$
\par

\begin{figure}
\centering
    \includegraphics[height = 7.25cm, width = 7.25cm, trim=3cm 3cm 3cm 3cm, clip]{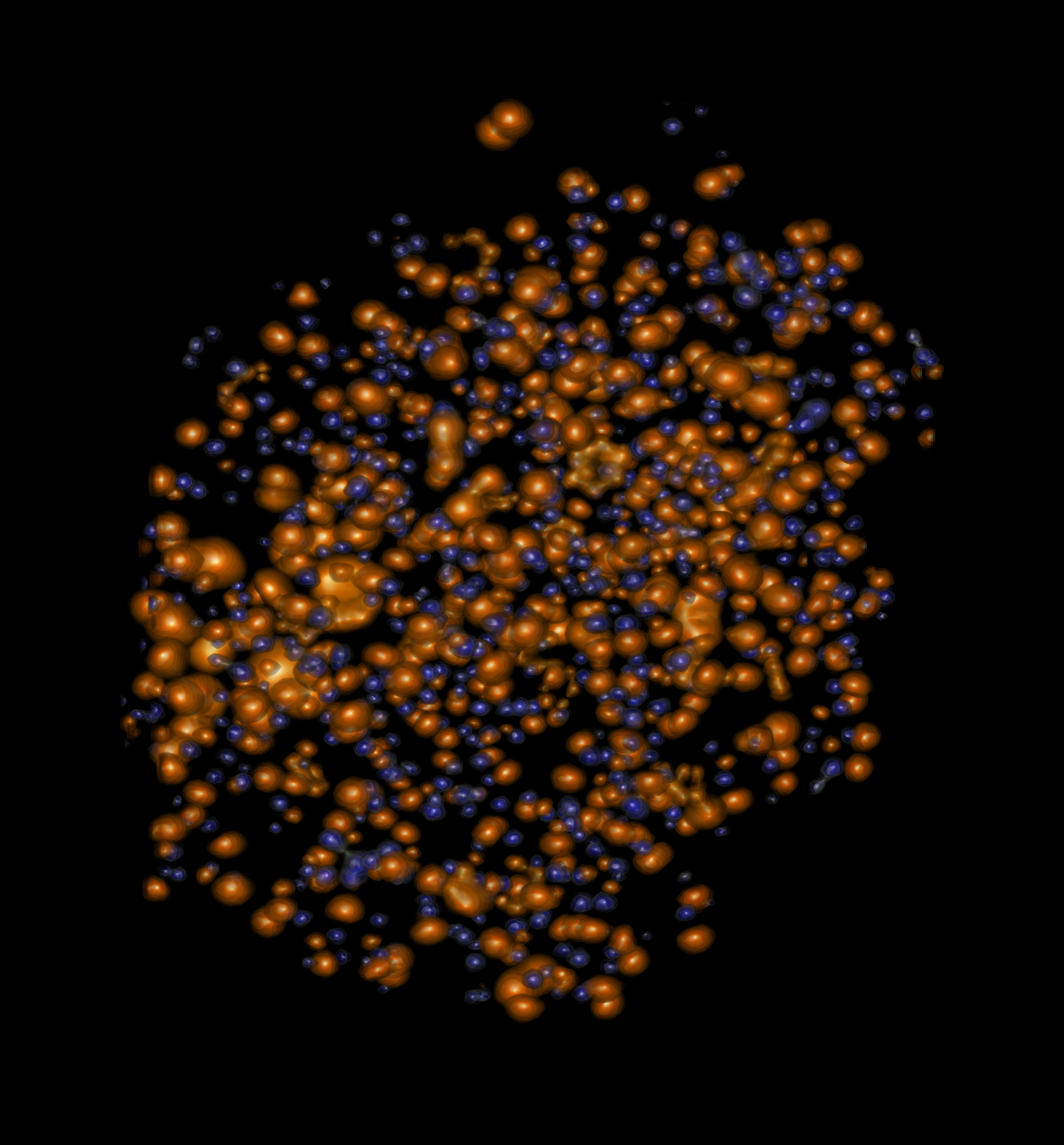}
    \caption{Non-linear electrostatic potential of Trypsin (PDB 1PPE).
    The potential fields where created with APBS \cite{Baker2001} rendered with  VolRover \cite{Bajaj2003,Bajaj2005}.  The positive
    and negative potential are rendered with blueish and orange/reddish colors  respectively.}
\label{figure:trypsin}
\end{figure}

We develop a framework to quantify the uncertainty of proposed binding sites given a stochastic manifold representation of the shape of
the receptor and ligand in solvent. This approach originates from  the fields
of scientific computing and numerical analysis.  More precisely,
these tools are adopted from the field of uncertainty quantification to assess the fitness of proposed \emph{rigid body} docking sites.  
To facilitate this process, we have developed a software tool that serves as a wrapper (plugin) for PIPER \cite{PIPER2006}, enabling the quantification of uncertainty in the predicted rigid body docking site. With the application of
Karhunen-Lo\`{e}ve Theorem on conformational molecular shape, we confirm the existence of a stochastic manifold representation of a conformational shape as an infinite linear combination of orthonormal functions with stochastic coefficients. This expansion is optimal compared to all other orthogonal representations in the sense of minimizing mean square error. A detailed discussion on the Karhunen-Lo\`{e}ve Theorem, essential for comprehending the functionality of our stochastic framework, is available in the next section.

\section{Mathematical Background}
\subsection{Karhunen-Lo\`{e}ve Theorem}

Due to the stochastic nature of the interactions between  solvent
and protein, we model the uncertainty in molecular shape as 
a random field (stochastic process). This representation is 
infinite dimensional, motivating the construction of accurate
finite dimensional noise models. The KL theorem provides an optimal finite dimensional model in a stochastic norm sense.
The Karhunen-Lo\`{e}ve Theorem  has been extensively used in model reduction, data analysis, signal processing, and many other fields. This expansion is also known as a \textit{proper orthogonal decomposition} and the methodology is also denoted as \textit{functional principal components}.\par
Let $D$ be a subset in $\mathbb{R}^{d}$. Define $L^2(D; \mathbb{R}^{q})$ to be a Hilbert space which consists of all the square integrable functions from  $D$ to $\mathbb{R}^{q}$ equipped with the inner product $(\cdot,\cdot)_{L^2(D;\R^q)}$ defined as $(f,g)_{ L^2(D;\R^q)} :=  \int_{D} f^{T}g\, \mbox{d} x$ for all $g,f \in L^2(D; \mathbb{R}^{q})$. 
Consider the random vector field $u : \Omega \to L^2(D;\mathbb{R}^q)$.
We define a suitable Hilbert space that incorporates the spatial and stochastic components. 
To this end, let $L^2_{\bbP}(\Omega;  L^2(D;\mathbb{R}^q))$
be a Bochner space that is equipped with the inner product  $(\cdot,\cdot)_{L^2_{\bbP}(\Omega;L^2(D; \mathbb{R}^q))}$
such that 
for all $w,v \in L^2_{\bbP}(\Omega;L^2(D;\mathbb{R}^q))$ we have that
\[
(w,v)_{L^2(\Omega;L^2(D; \mathbb{R}^q))}
:= \eset{(w,v)} 
:= \int_{\Omega} (w,v)\, \mbox{d} \pset. 
\]
Note that the corresponding norm $\| v \|_{L^2_{\bbP}(\Omega;  L^2(D;\mathbb{R}^q))}
$ for all  $v \in L^2(\Omega;  L^2(D;\mathbb{R}^q))$
of the Bochner space is defined as
\[
\| v \|_{L^2(\Omega;  L^2(D;\mathbb{R}^d))} := (v,v)^{\frac{1}{2}}_{L^2(\Omega;L^2(D; \mathbb{R}^q))}.
\]

\begin{remark}
The definition of the Bochner space might appear somewhat abstract. However, this space will be important
since we assume that the random fluctuations of the molecule shape are described     
by a random vector field $u \in L^2(\Omega;L^2(D; \mathbb{R}^q))$. This space allows us to
construct finite dimensional representations of the random fluctuations of the protein by
using KL expansions.
\end{remark}


\begin{definition}~
\begin{enumerate}
\item  For all $u(\bv) \in L^2(\Omega;L^2(D;\mathbb{R}^q))$, where $u({\bv}):= \left[u_1(\bv,\omega),u_2(\bv,\omega),...,u_d(\bv,\omega)\right]^T$, let 
\[
\eset{u_i(\bv,\omega)} := \int_{\Omega} u_i(\bv,\omega)\,\emph{d} \pset,
\]
 for $i = 1,2,...,d$.\par

\item  For all $u \in L^{2}\left(\Omega ; L^{2}\left(D ; \mathbb{R}^{q}\right)\right)$, 
\[
\begin{split}
&\operatorname{Cov}\left(u_{i}(\mathbf{v}), u_{j}(\mathbf{y})\right) := \mathbb{E}\left[\left(u_{i}(\mathbf{v})-\mathbb{E}\left[u_{i}(\mathbf{v})\right]\right) \right.\\
&\left.\left(u_{j}(\mathbf{y})\right.\right.\left.\left.-\mathbb{E}\left[u_{j}(\mathbf{y})\right]\right)\right]      
\end{split}
\]

for $i, j=1, ..., d$. 
Denote the covariance matrix function of $u$  between index $\bf{v}$ and index $\bf{y}$ as 
\begin{align}
R_u(\bf{v},\bf{y}) &:= \operatorname{Cov}(\bf{u}(\mathbf{v}), \bf{u}(\mathbf{y})) 
\end{align}
\item Associate to $R_u$ a linear operator $T_{R_u}$ defined  in the following way:
\begin{equation}\label{tranfdc}
       T_{R_u} :L^{2}(D)  \to L^{2}(D): f \to T_{R_u} = \int_{D} R_u(\bs,\cdot) f(\bs)d\bs  
\end{equation}

\begin{equation}\label{Fredholm}
     \int_{D} R_u(\bs,\bv) \phi_k(\bs)d\bs= \lambda_k \phi_k(\bv)
\end{equation}
where $\{ \phi_{k}, k \in \mathbb{N} \}$ are orthonormal eigenfunctions of $T_{R_u}$ in $L^2(D)$ with respect to eiganvalues $\{\lambda_k, k \in \mathbb{N}\}$.
\end{enumerate}
\label{defn1}
\end{definition}

The following theorem shows that any random field $u \in L_{\bbP}^2(\Omega, L^2(D,\bbR^q))$ can be represented as an infinite sum in terms of  eigenvalues, eigenfunctions $\{(\lambda_k, \phi_k)\}_{k \in \bbN}$ and random components. 
\begin{theorem}[Multivariate Karhunen-Lo\`{e}ve expansion]\label{KL}
Let $ u(\bv,\omega) $ be a zero mean vector process in $L^2(\Omega; L^2(D;\mathbb{R}^q))$.  Then $u(\bv,\omega)$ admits the following representation:
\begin{equation}\label{kl}
    u(\bv,\omega) = \sum_{k = 1}^{\infty} Z_k(\omega) \phi_{k}(\bv),
\end{equation}
 where the convergence is in $\|\cdot\|_{L^2_{\bbP}(\Omega;L^2(D;\mathbb{R}^q))}$, and 
 \begin{equation}\label{random_coefficient}
   Z_k(\omega) = \int_{D} 
   u(\bv,\omega)^{T}\phi_k(\bv)\emph{d}\bv  
 \end{equation}

\par Furthermore, ${Z_k}$ are uncorrelated with mean zero and variance $\lambda_k$.
\end{theorem}

\begin{remark}
 The general case of a process $u_\bv$ that is not centered can be brought back to the case of a centered process by considering $u_{\bv} - \mathop{\mathbb{E}}(u_{\bv})$, which is centered. \par
 \end{remark}

 A key feature of the KL Theorem is that the truncated expansion is optimal in the 
sense that among all finite dimensional orthonormal basis approximations
it minimizes the total mean square error.

\begin{proposition}

Let $\mcB = \{\psi_{i}(\bv,\omega) \}_{i \in \bbN}$ be a complete orthonormal basis of 
$L^2_{\bbP}(\Omega;L^2(D; \mathbb{R}^d))$, and $\mcB^{p} = \{\psi_{i}(\bv,\omega) \}_{i=1,\dots,p}$ be a collection of $p$ basis functions in $\mcB$.
Let $\tilde u_p$ be the orthogonal projection of $u(\bv,\omega)$ (approximation) onto the 
finite dimensional subspace with the following orthonormal basis functions, e.g.
\[
\tilde{u}_p(\bv, w)= \sum_{i=1}^{p} \left(\int_{D} \int_{\Omega} u(\bv,\omega) \psi_{i}(\bv,\omega)\,\emph{d}\bv
\emph{d} \bbP
\right) \psi_{i} (\bv,\omega)
\]

Denote $\mcC$ to be the set of all complete orthonormal bases of $L^2_{\bbP}(\Omega;L^2(D; \mathbb{R}^d))$.
For any basis $\tilde \mcB \in \mcC$ let $\tilde \mcB^{p}$ be the collection of any $p$ basis functions in $\tilde \mcB$
and
\[
\mcC^{p} = \{\tilde \mcB^{p}\,|\,\tilde \mcB \in \mcC \}.
\]
Then
\[
\begin{split}
&\arginf_{\mcC^{p}} \int_{D} \eset{\text{Err}_{p}^{2}(\bv)}
\, \emph{d} \bv = \{\frac{\phi_{1}(\bv)Z_1(\omega)}{\sqrt{\lambda_1}},  \frac{\phi_{2}(\bv)Z_2(\omega)}{\sqrt{\lambda_2}},\\ &\dots,\frac{\phi_{p}(\bv)Z_p(\omega)}{\sqrt{\lambda_p}}\}
\end{split}
\]

where
$$
 \text{Err}_{p}(\bv) := u - \tilde{u}_p = \sum_{i\geq p+1} \alpha_i
\psi_{i} (\bv,\omega),
$$
with each coefficient $\alpha_i$ is given by
\[
\alpha_i = \int_{D} \int_{\Omega} u(\bv,\omega) \psi_{i}(\bv,\omega)\,\emph{d}\bv
\emph{d} \bbP
\]
for $i \geq p + 1$.

\end{proposition}
\begin{remark}
    To apply the Karhunen-Lo\`{e}ve expansion, we only need to have eigenfunctions $\phi_k$ and random coefficients $Z_k$. It is not hard to obtain the former as long as we can construct a convariance matrix based on data and apply the Method of Snapshots \cite{Sirovich} to estimate eigenfunctions empirically. However, it can be infeasible to estimate the true probability distribution of the random field and hence the true random coefficients $Z_k$ due to high dimensions, especially in general, $Z_k$ are only uncorrelated rather than independent.
\end{remark}
\begin{remark}
   In practice, the Karhunen-Lo\`{e}ve expansion is truncated to a finite number of terms. Consider a $d \times d$ covariance matrix, where $d$ denotes the feature dimension. 
High feature dimension leads to a large covariance matrix, where solving the
eigen decomposition problem may be infeasible.
The Method of Snapshots derived by Sirovich \cite{Sirovich} can potentially reduce the high dimensional problem of finding eigenfunctions of a continuous convariance function $R_u(\bs,\bv)$, to an eigen-decomposition problem for a finite-dimensional matrix.  This is done by taking snapshots (samples) at discrete times, where the number of snapshots is usually much smaller than dimension $d$, hence making the 
solution of the eigen-decomposition problem more affordable (see Appendix \ref{snapshots} for more  details). 
\end{remark}

 It can be shown that the Karhunen-Lo\`{e}ve expansion preserves the covariance structure of observations no matter what distribution the random coefficients follow.
\begin{proposition}\label{prop0}
    Assume that $\tilde{Z}_k(\omega)$ for all $k \in \mathbb{N}$ are orthonormal in $L^2(\Omega;  L^2(D;\mathbb{R}^q))$  of mean 
    zero, and of variance $\lambda_k$. The Karhunen-Lo\`{e}ve expansion of zero mean vector process $u(\bv)$ with $\tilde{Z}_k$ as random coefficients will be:
    $$\tilde{u}(\mathbf{v}, \omega)=\sum_{k \in \mathbb{N}} \phi_k(\bv)\tilde{Z}_k(\omega)$$
    
  The covariance function
 $\operatorname{Cov}(\tilde{v}(\mathbf{x}, \omega), \tilde{v}(\mathbf{y}, \omega))=\sum_{k \in \mathbb{N}} \lambda_k \phi_k(\mathbf{x}) \phi_k(\mathbf{y}) =\operatorname{Cov}(v(\mathbf{x}, \omega), v(\mathbf{y}, \omega))$. In other words, no matter what distribution $\tilde{Z}_k$ follows, the covariance structure of the new random field formed based on these random coefficients remains consistent with that of the original.
\end{proposition}

Furthermore, if the random field is Gaussian, random coefficients $Z_k$ are not just uncorrelated.
\begin{proposition}
If the process $u(\bv) \in  L^2(\Omega:  L^2(D;\mathbb{R}^d))$ is Gaussian, then the random variables $Z_k \sim \mcN(0,\lambda_k)$ are normal and identically independent distributed.
\end{proposition}
Consequently, we can assume that the random vector field $u(\bv, \omega)$ follows a Gaussian distribution, which implies the independence of random coefficients. As already demonstrated in Proposition \ref{prop0}, regardless of the distribution assigned to the stochastic manifold, the covariance structure for atomic coordinates remains consistent. This consistency is ensured by Mercer's theorem, which affirms that a symmetric, positive-definite matrix can be uniquely expressed as a sum of a convergent sequence, composed solely of eigenvalues and eigenfunctions (as shown in Appendix, equation \eqref{mercer}). This representation thus retains the original molecular structural information \cite{juliostochastic2022}.

\begin{remark} For reasons of generality, the probability spaces are defined with respect to an extensive set of outcomes $\Omega$. 
However,
under certain conditions, the probability measure $\pset:\Omega\rightarrow 
\R^{+}$  can be associated with the probability density function
$\varrho(\bx):\R^q \rightarrow \R^{+}$, for some $q \in \mathbb{N}^{+}$\cite{Castrillon2016}. 
Let $\bz = \left(z_1(\omega),z_2(\omega),...,z_q(\omega)\right)$ be a q-valued random vector where each random variable $z_k(\omega)$ is defined as in Theorem \ref{KL}. Note that the random vector $\bz$ is a function from $\Omega$ to $\mathbb{R}^{q}$.
This 
gives us the more familiar form of expectation:
\[
 \mathbb{E}(u(\bv,\omega)) := \int_{\Omega} u(\bv,\omega)\,\emph{d} \pset =\int_{\mathbb{R}^q} u(\bv,\bz) \,\varrho(\bz)\mbox{\mbox{\emph{d}}}\bz.
\] 
By our assumption, the probability density function $\varrho$ is known to be the Gaussian probability density function. For convenience, in the subsequent sections, we will replace all $\omega$ with $\bz(\omega)$.

\end{remark}

\section{Problem Formation}
\subsection{Protein Stochastic Model}

\begin{figure}
\centering
\begin{tikzpicture}[scale=0.9]
\draw[dashed,fill=black!30!green, opacity=0.2] (0,0) to
[closed, curve through={(1,.5) .. (2,0) .. (3,.5)}] (4,0);
\node at (2,-0.9) {$\bx_1$};
\draw[fill = black] (2.3,-1.1) circle (2pt);
\node at (0.7,-2) {$\bx_2$};\draw[fill = black] (1,-2.2) circle (2pt);
\node at (2,-1.5) {$\bx_{k}$};
\draw[fill = black] (2.3,-1.7) circle (2pt);
\node at (3,-0.3) {$\bx_{n}$};
\draw[fill = black] (3.1,-0.6) circle (2pt);
\node at (2.7,-2.5) {$\epsilon_{{\mcA}(\vartheta)}(\bx,\vartheta)$};
\node at (3.5,-3.3) {$\epsilon_{S}(\bx,\vartheta)$};
\node at (0.7,-.3) {${\mcA}(\vartheta)$};
\node at (-0.25,0.5) {$\Upsilon_{S}(\vartheta)$};
\node at (4.5,0.5) {$\partial \mcD$};
\draw[black!60!green]  (-1,-3.75) rectangle (5,1);
\end{tikzpicture}
\caption{Two dimensional representation of receptor and solvent.}
\label{SPDEs:fig1}
\end{figure}
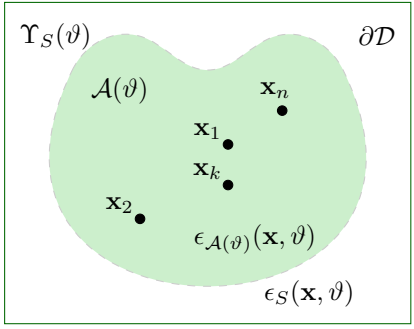

The uncertainties in the molecular conformation can be
propagated to the receptor and ligand molecular maps and eventually to
the energy function as
$$\E(\balpha,\blambda,\omega):= \sum_{p=1}^{P}
\int_{\R^{3}} \overline{R_p(\bx,\omega)} T(\blambda)
  D(\balpha)L_{p}(\bx,\omega) \mbox{d}\bx.$$
A starting point is to determine the docking site $(\balpha,\blambda)$ that minimizes mean log 
$\esets{\log \E(\balpha,\blambda,\omega)}.$
Nevertheless, a large standard deviation suggests significant variability in the energy function with respect to random conformations. Under such circumstances, it may be more suitable to identify potential candidates for the docking site
$(\balpha,\blambda)$
which exhibit both small means and small standard deviations.
The mean and second moment of energy are defined as:
\begin{equation}
\begin{split}
       &\eset{\log(\E(\balpha,\blambda, \omega) + \delta)}\\ &= \int_{\Omega} \log(\E(\balpha,\blambda, \omega) + \delta) \rho(\omega) \id \omega  
       \end{split}
   \label{eq:mean}
\end{equation}
 and

\begin{equation}
    \begin{split}
        &\eset{\left(\log(\E(\balpha,\blambda, \omega) + \delta)\right)^2}  \\
        &= \int_{\Omega} \left(\log(\E(\balpha,\blambda, \omega) + \delta)\right)^2 \rho(\omega) \id \omega.
    \end{split}
    \label{eq:second moment}
\end{equation}

A better choice is to weight the interaction energy using the standard deviation using the
stochastic optimization
\begin{equation}
\begin{split}
(\bm{\alpha}^*,\bm{\lambda}^*) &:= \argmin_{\balpha,\blambda} \left( \esets{ \log(\E(\balpha,\blambda,\omega) + \delta)} \right. \\
&\quad \left. + \beta \, \mbox{SD}[\log \E(\balpha,\blambda,\omega)] \right),
\end{split}
\label{description:eqn1}
\end{equation}
for some user-given parameter $\beta >0$ and $\delta > 0 $, ensuring 
 that $\E(\balpha,\blambda,\omega) + \delta > 1$. This optimization
will look for the docking configuration with a small mean and small dispersion from the stochastic molecular
conformation. Solving it involves computing the mean and standard deviation of the energy
function.
For many cases, the evaluation of each rotational and translational search $(\bm{\alpha}, \bm{\lambda})$ involves a function $\E(\bm{\alpha}, \bm{\lambda}, \omega)$ that is high-dimensional, non-Gaussian, and non-linear with respect to the stochastic parameter $\omega$. Each rotational and
translational configuration $(\balpha,\blambda)$ will correspond to
computing 
the mean and SD of the high dimensional function
$\log (\E(\balpha,\blambda,\omega) + \delta)$. 
If we model $\E(\balpha, \blambda, \omega): \Omega \rightarrow \mathbb{R}^6$ with $N$ stochastic dimensions, the total number of dimensions of the domain of $\E(\balpha, \blambda, \omega): \mathbb{R}^6 \times \Omega \rightarrow \mathbb{R}$ is $N+6$. Consequently, even with a relatively small number of dimensions $N$, stochastic optimization becomes intractable. 

To manage this complexity, we employ an uncertainty decision tree to eliminate configurations of high uncertainty which are unlikely to be the true docking site.
This process is based on determining the level of uncertainty through comparing mean and variance. Configurations with both low mean and low variance can be considered  promising candidates for optimal docking sites. Conversely, configurations with high mean or variance are considered highly uncertain and thus unlikely to be the optimal docking site.
More details are shown in \lcol{Figure} \ref{figure:decisiontree}. 
By leveraging the decision tree, we can carefully select a list of plausible docking site candidates by excluding those configurations of  high uncertainty.




\begin{remark}
 Note that, due to local low regularity and high dimensions,
 comprehensively capturing all uncertainty within the stochastic molecular manifold is computationally challenging. However, identifying sites with exceptionally high uncertainty remains feasible through the calculation of the mean and variance of energy. If these metrics indicate high uncertainty, we can have high confidence in these findings. Conversely, if both mean and variance are low, this does not conclusively indicate low uncertainty at the current stage, due to potential computational inaccuracies arising from truncation and integration errors. For a detailed analysis of the errors and challenges faced by this framework,
 please refer to the Discussion
 \ref{Discussion}.
\end{remark}
\begin{remark}
Criteria based on means and variances of energies, $C_{\mu}$ and $C_{\sigma^2}$, have to be established to assess whether a configuration exhibits high uncertainty. One approach involves selecting high quantiles (e.g., the 0.8 quantile) of the empirical distribution of energy means and variances to form the respective criteria.
 \label{remark:uncertainty}
\end{remark}

\begin{figure}
\centering
\begin{tikzpicture}
  [
    sibling distance        = 7em,
    level distance          = 11em,
    edge from parent/.style = {draw, -latex},
    every node/.style       = {font=\footnotesize},
  ] 
  \node at (-3,0) {\includegraphics[scale=0.1]{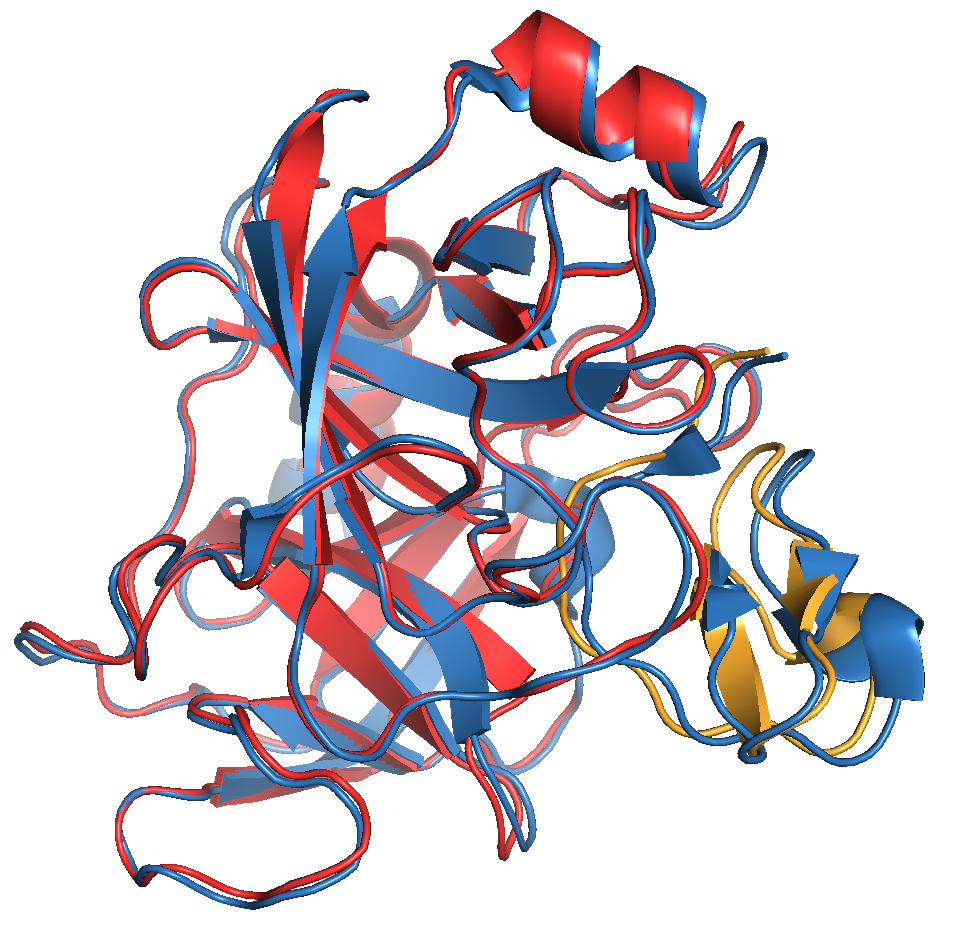}};
  \node [root] {Docking site \\ Candidate}
    child { node [env] {Unqualified \\ Candidate}
      edge from parent node  [left] [align=center]{$\eset{\log(\E(\balpha, \blambda, \omega) + \delta)}$ \\ $> C_{\mu}$} }
    child { node [env,circle,minimum size = 7mm] {}
      child { node [env] {Unqualified \\ Candidate}
        edge from parent node [left][align=center] {$\text{Var}[\log(\E(\balpha, \blambda, \omega) + \delta)]$\\$> C_{\sigma^2}$} }       
           child { node [env] {Potential \\ Candidate}
              edge from parent node [right][align=center]
                {$\text{Var}[\log(\E(\balpha, \blambda, \omega) + \delta)]$\\$\leq C_{\sigma^2}$}}
              edge from parent node [right] [align=center]{$\eset{\log(\E(\balpha, \blambda, \omega) + \delta)}$ \\ $\leq C_{\mu}$} };
\end{tikzpicture}

\caption{Decision Tree for Docking Site Candidate Selection:
Step 1: Compare the variance of energy, if mean is greater than the predefined criterion $C_{\mu}$, then the corresponding configuration is determined to be of high uncertainty, and should be filtered out. Step 2:  Among the configurations with mean less than $C_{\mu}$, those with an energy variance exceeding $C_{\sigma^2}$, the variance criterion, are discarded. Conversely, configurations with a variance below $C_{\sigma^2}$ are retained. By removing those configurations which already be determined as impossible candidates of high uncertainty, the final result obtained is a list of potential candidates which is awaited for further study in higher dimensional stochastic molecular conformational space to identify the optimal docking site.
}
\label{figure:decisiontree}
\end{figure}

\subsection{The approach}
In this section, we outline the methodology for establishing the stochastic framework of the molecular dynamic manifold. We start by assuming that the conformation manifold of either the receptor or ligand behaves as a Gaussian random field. This field is then truncated by drawing a circle centered on the mean, with a radius extending to three standard deviations. 


The first step involves leveraging the Karhunen-Loève expansion, as previously described, on the atom coordinates of the receptor or ligand to derive the stochastic atom coordinates. The random vector fields $u_R$ and $u_L$ represent the three-dimensional stochastic coordinates within the molecular manifolds of the receptor and ligand, respectively, expressed as $u_R(\bv,\omega)= [x_R(\bv,\omega), y_R(\bv,\omega), z_R(\bv,\omega)]^T$ for the receptor and $u_L(\bv,\omega)= [x_L(\bv,\omega), y_L(\bv,\omega), z_L(\bv,\omega)]^T$ for the ligand.

The truncated representations of these stochastic coordinates are as follows:
\[
\begin{split}
    F_{R}(\bv,\omega) &= \mathbb{E}(u_R)+\sum_{k = 1}^{N_R} Z^{R}_k(\omega)\phi^{R}_{k}(\bv)\\ &= \mathbb{E}(u_R)+\sum_{k = 1}^{N_R}\sqrt{\lambda_{k}^R} \phi^R_{k}(\bv)Z^{R}_{k}(\omega) \\
\end{split}
\]
where
\[
\begin{split}
Z^{R}_k(\omega) &= \int_{D} u_{R}(\bv,\omega)^{T}\phi^{R}_k(\bv)d\bv 
\end{split}
\]
and
\[
\begin{split}
    F_{L}(\bv,\omega) &= \mathbb{E}(u_L)+\sum_{k = 1}^{N_L} Z^{L}_k(\omega)  \phi^{L}_{k}(\bv) \\
    &=\mathbb{E}(u_L)+ \sum_{k = 1}^{N_L}\sqrt{\lambda_{k}^L} \phi^L_{k}(\bv)Z^{L}_{K}(\omega)\\
    \end{split}
    \]
where
    \[
    \begin{split}
    Z^{L}_k(\omega) &= \int_{D} u_{L}(\bv,\omega)^{T}\phi^{L}_k(\bv)d\bv.
\end{split}
\]
In the above equations, $F_{R}$ and $F_{L}$ represent the truncated stochastic approximations of $u_R$ and $u_L$, respectively. $N_R$ and $N_L$ denote the dimensions of truncated stochastic spaces for the receptor and ligand. $(\lambda_k^R,\phi_k^R(\bv))$ and $(\lambda_k^L,\phi_k^L(\bv))$ represent the eigenpairs for the receptor and ligand, with the eigenvalue coefficients $\lambda_k^R$ and $\lambda_k^L$ organized in a descending sequence by $k$. And $Z_k$'s are independent, zero-mean, unit variance normal random variables.
\par
With these stochastic coordinates, it is feasible to generate PDB  files for stochastic receptor/ligand domains. This enables the generation of stochastic shape conformation realizations for receptor/ligand that are independent of specific molecular mapping techniques, as these realizations are fundamentally based on atomic coordinates. Subsequently, by employing docking software such as PIPER, one can obtain the energies associated with these stochastic domains. These energy calculations enable the computation of statistical measures, enhancing our understanding of molecular interactions within this stochastic framework.

\begin{figure*}[h]
     \centering
\includegraphics[scale = 0.7,width=\linewidth]{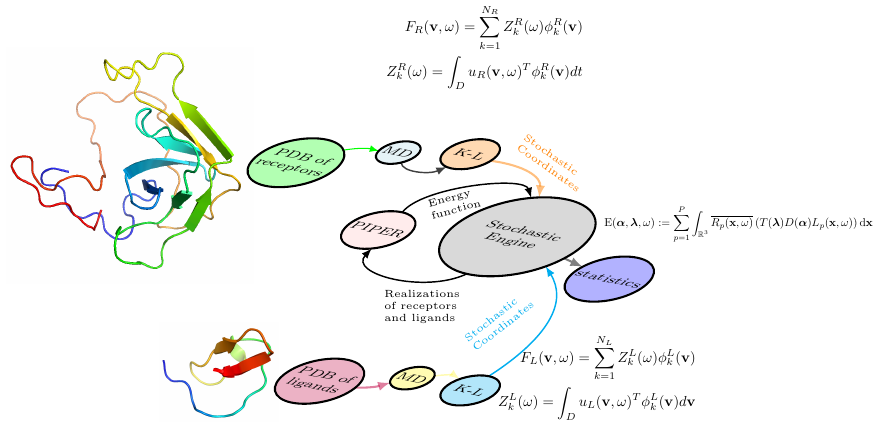}

       \caption{The flow diagram of the process of constructing the stochastic framework for the molecular dynamic manifold. Step 1: Obtain  realizations of receptor/ligand by inputting PDB files into molecular dynamic software such as NAMD. Step 2: Apply Karhunen-Lo\`{e}ve expansion to generate stochastic 3-dimensional coordinates. Step 3: Input stochastic coordinates into Stochastic Engine to generate realizations of stochastic receptor/ligand manifolds. Step 4: Use rigid protein docking program such as PIPER to evaluate interaction energies of the stochastic molecular conformation. Step 5: Compute statistics (mean and standard deviation) to look for the optimal docking site or promising candidates of the optimal docking site. }
    \label{pipeline}
\end{figure*}
\begin{figure}
\centering
\begin{tikzpicture}
\node at (0,0) {\includegraphics[scale = 0.4 ,trim= 2cm 6cm -3cm 9cm,clip = true]{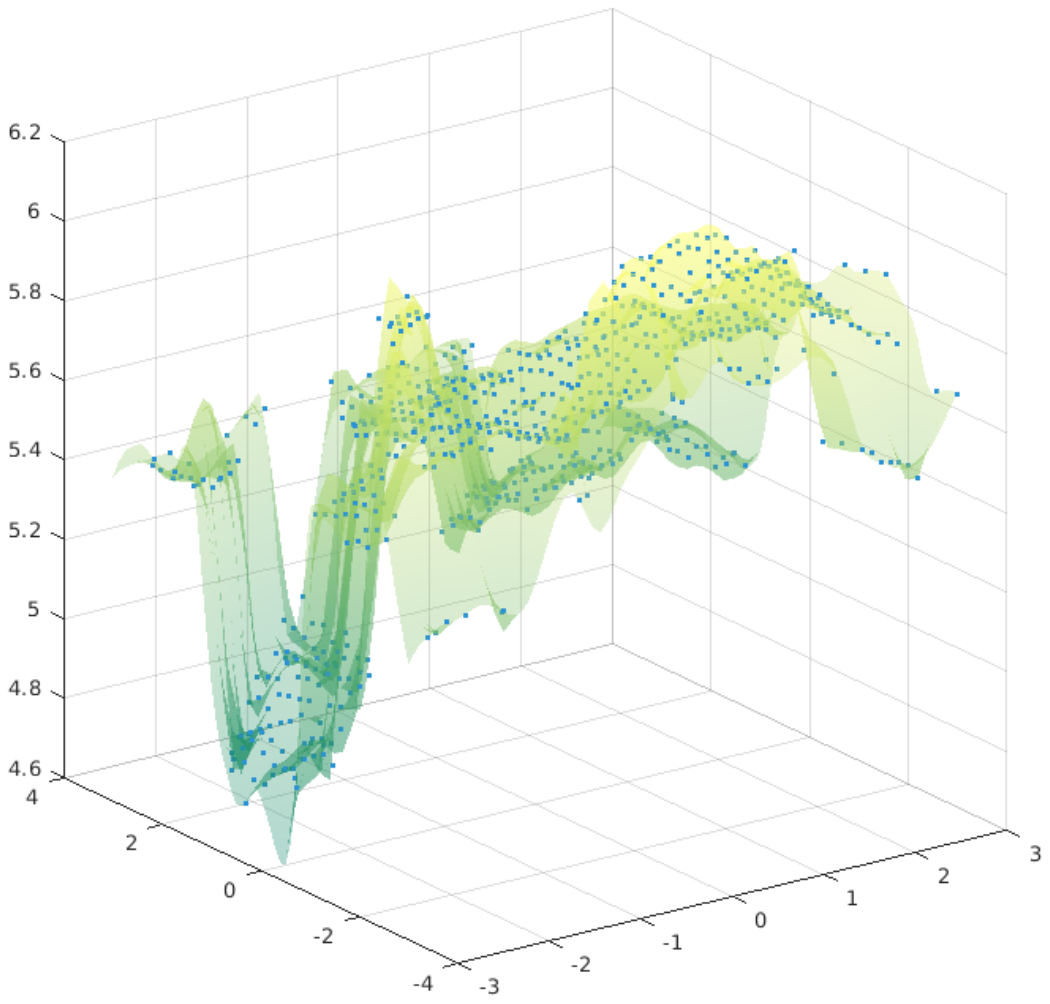}};
\node at (-3,-2.5) {$Z_1(\omega)$};
\node at (1,-2.5) {$Z_2(\omega)$};
\node at (-5.1,1) [rotate=90] {$\log(\E(\balpha,\blambda,\omega) + \delta)$};
\end{tikzpicture}
\caption{Energy profile from PIPER with 
respect to 2 dimensional stochastic deformation for Trypsin. As the figure shows, there is a relatively sharp discontinuity on the energy manifold, which prompts us to look for a more feasible method compared to traditional Gaussian quadrature methods to deal with the daunting computation task with reasonable accuracy and cost.}
\label{SPDEs:fig2}
\end{figure}

\subsection{Computation of statistics}

The objectives of this calculation are delineated in equations \eqref{eq:mean}  and \eqref{eq:second moment}, which represent the mean and the second moment of energy, respectively. Common numerical methods for approximating integrals include quadrature methods such as Simpson's method \cite{simpson}, Gaussian quadrature \cite{GQ}, and etc.
Quadrature methods aim to approximate definite integrals by evaluating the function at designated points, subsequently applying weights, and summing the results. For example, let $f:S \to \mathbb{R}$ where $S:= [a,b]^n \text{ and }a,b \in \mathbb{R}$,

$$\int_{S} f(\bx) \mbox{d} \bx \approx \sum_{i=1}^n W_i f({\bx}_i)  = \mathcal{R}(f(\bx))$$
 where $\{{\bx}_i\}_{i=1}^{n}$ are the quadrature points.  Here 
 $\{W_i\}_{i=1}^n$ denote corresponding weights.  Both of the above depend on the choice of quadrature rule, and $\mathcal{R}(f(\bx))$ denotes the integration interpolant of function $f(\bx)$.
 
 \par
Considering the constraints of the protein energy manifold, direct evaluation of these integrals is non-trivial owing to the finite spectrum of known discrete values. Moreover, the designated quadrature points in these methods present additional challenges, especially when values at these locales are difficult to ascertain. It's important to note that if there's any irregularity at just one point within the range we're looking at, the accuracy of our approximation will suffer.

\lcol{Figure} \ref{SPDEs:fig2} displays the energy profile of the electrostatic field for the Trypsin protein, utilizing the Karhunen-Lo\'{e}ve approximation of a two-dimensional stochastic deformation from Piper.  As observed, the profile appears to be locally smooth, but not globally. 
This motivates us to seek a more accurate representation of quantity of interest (equations \eqref{eq:mean} and \eqref{eq:second moment})  that can capture both high and low regularity components of the energy domain in a quadrature form. We now consider a representation based on radial basis functions (RBF)
together with polynomial interpolation:

$$  {\mathcal{R}}[f(\bx)] =
\sum_{i=1}^{p} \mathbf{e}[i]m_{i}(\bx)
+
\sum\limits_{j=1}^{\eta} {\mathbf{v}[j] } \phi
(\bx,{\bx}_j ).
$$
The first term is the polynomial interpolation with lagrange basis $m_i(\bx)$ and weights $\mathbf{e}[i] 
 \in \mathbb{R},i=1,\dots,p$, capturing the region of high regularity. The second term is the radial basis interpolation with isotropic kernel function $\phi:\mathbb{R}^n \times \mathbb{R}^n\rightarrow \mathbb{R}$, nodes $\{\bx_j\}_{j=1}^{\eta}$, and weight $\mathbf{v} \in \mathbb{R}^{\eta} $ capturing the region of low regularity. 

Note that our manifold is stochastic, with suitable choice of weights and quadrature nodes $\omega_j \in \Omega$, the interpolation representation of equation (\ref{eq:mean}) will be the following:


\begin{equation*}
\begin{split}
\eset{\log(\E(\balpha, \blambda, \omega) + \delta)} &\approx \mathcal{R}(\eset{\log(\E(\balpha, \blambda, \omega) + \delta)})\\ 
&= \sum_{i=1}^{p} \mathbf{e}[i]m_{i}(\omega)
+ \sum\limits_{j=1}^{\eta} {\mathbf{v}[j] } \phi(\omega,\omega_j ).
\end{split}
\end{equation*}


We introduce a quadrature scheme that employs both polynomial and radial basis functions with a Gaussian measure (in preparation). Leveraging the symmetry of the Gaussian measure and the centrality of quadrature points, this approach offers high accuracy at each point—even in regions of low smoothness.

\section{Experiments and results}
We test the effects of solvent uncertainty on the predicted docking site
of bovine beta trypsin (chain E, Receptor) with the CMTI-I trypsin inhibitor (chain I, Ligand) 
from squash. Initially we run the rigid body docking code Piper \cite{Kozakov2010} and the
70,000 rotations $\balpha$ (and corresponding translations $\blambda$) are ordered from
lowest energy (best fit) to highest (worst fit). One hundred
realizations of the receptor are generated using  Scalable Molecular 
Dynamics NAMD \cite{Phillips2020} software. From these realizations an 
optimal truncated Karhunen-Lo\'{e}ve (KL) \cite{loeve1978,Harbrecht2016,Schwab2006}
stochastic model of the receptor domain is formed:
\begin{equation}
\mcA(\bx,\omega) 
\approx \eset{\mcA(\bx,\omega) } +
\sum_{n=1}^{N_{\mcA}} \sqrt{\lambda_{n}} \phi_{n}(\bx)
Z_{n}(\omega).
\label{setup:eqn1}
\end{equation}
It is assumed that
the random field of the receptor domain $\mcA(\bx,\omega)$ is a Gaussian process. 
The eigenfunctions $\{\phi_n(\bx)\}_{n=1}^{N_{\mcA}}$ 
can be estimated empirically using
the Method of Snapshots \cite{Castrillon2002}.
The eigenvalue coefficients $\lambda_{n} \in \R$ are monotonically decreasing
with respect to $n$ and $Z_n$ are independent zero mean, unit variance 
Normal random variables. For this experiment $N_{\mcA}$ is set to be 2.

In \lcol{Figure} \ref{preliminarynumericalresults:fig1}, the top line shows
$\eset{\log (\E(\balpha,\blambda,\omega)+ \delta)}$ while the bottom line represents
$\mbox{SD}{[\log (\E(\balpha,\blambda,\omega)]}$. The first 100 sorted rotations are plotted in 
\lcol{Figure} \ref{preliminarynumericalresults:fig1} (a) and sorted rotations from 69,050 to 60,150 in 
\lcol{Figure} \ref{preliminarynumericalresults:fig1} (b). Notice that the first 100 rotations generally exhibit
a lower mean and significantly smaller SD. The same pattern was observed 
for all the sorted rotations. Thus under receptor stochastic deformation the rigid body
docking site predictions have low uncertainty and are consistent with the actual docking site. Low uncertainty also
shows that this binding site is robust towards stochastic domain deformations of the receptor.

\begin{figure*}[htb]
  \centering
  \begin{tikzpicture}
  \begin{scope}[scale = 1, every node/.style={scale=1}]]
  \node at (0,0)  {\includegraphics[ width=2.5in, height=2.25in,trim=1.9cm 4.50cm 0.75cm 5cm, clip]{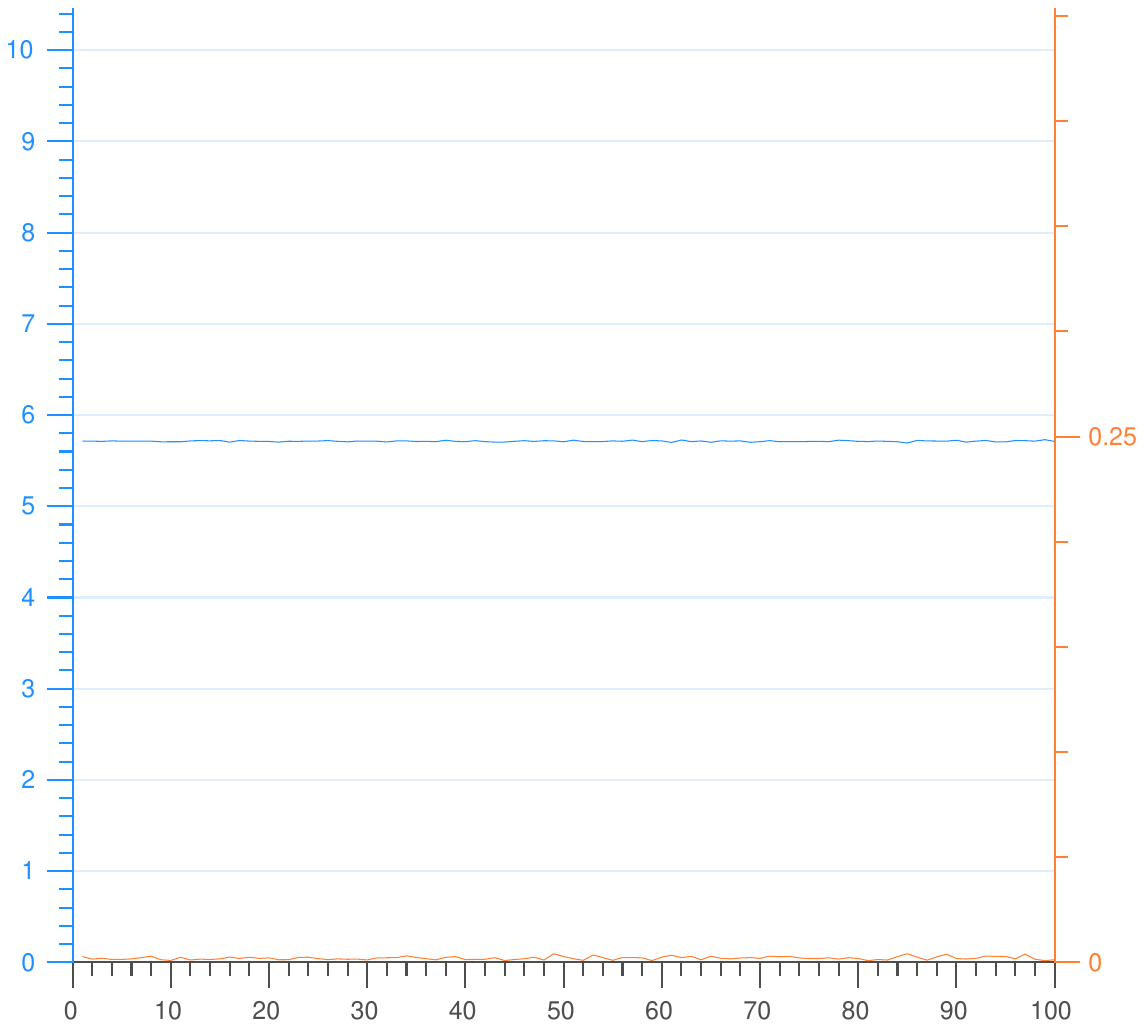}
                   \includegraphics[width=2.5in,height=2.25in,trim=1.9cm 4.50cm 0.75cm 5cm, clip]{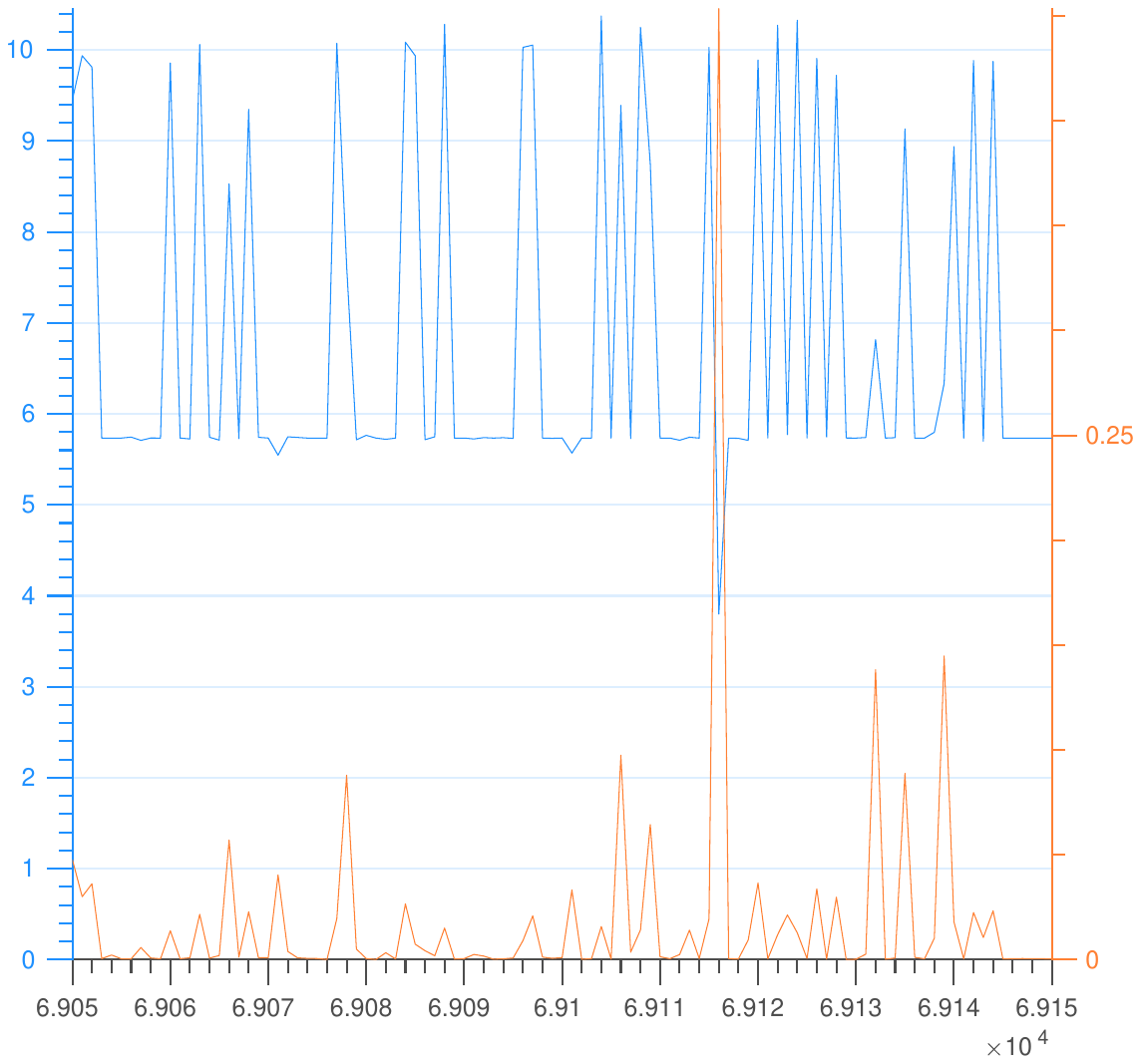} 
  };
  \node at (-5.65,2.75)  {(a)};
  \node at (0.75,2.75)  {(b)};

  
  \node at (-5,0.75) [right]  {$\eset{\log (\E(\balpha,\blambda,\omega)+ \delta)}$};
  \node at (-4.5,-3) [right]  {Rotation Number};
  
  \node at (0.5,-1.0) [right]  {$\mbox{SD}{[\log (\E(\balpha,\blambda,\omega)]}$};
  
  \end{scope}
\end{tikzpicture}
\caption{\corb{Trypsin energy manifold uncertainty measures with respect to 2-dimensional stochastic receptor deformations. 
The top line is the mean and the bottom the SD.  The corresponding scale axis are given to the left and the right.
(a) First 100 top sorted rotations. Note  flatness of the curve; in general mean and SD are more optimal.
(b) Sorted rotations from 69,050 to 69,150.}}
\label{preliminarynumericalresults:fig1}
\end{figure*} \par

When $N_{\mathcal{A}} = 3$, for the identical docking sites, the observations reveal that the predictions for rigid docking sites across the first 100 ranked rotations are significantly rougher and less stable than those derived from the two-dimensional model. This instability could be attributed to more intense collisions and fluctuations within the three-dimensional protein dynamic manifold.
To reduce noise, we scale the data after subtracting the mean. \lcol{Figure} \ref{fig:3dstats} shows that, after scaling,  the rigid docking
site predictions under 3-dimensional stochastic deformation
have low uncertainty, consistent with those observed in the case of two-dimensional stochastic deformation. However, as rotation number increases (\lcol{Figure} \ref{fig:3dstats} (c)), both the mean and standard deviation exhibit increased oscillation.  An alternative viable method could involve acquiring representations of receptors and ligands in three dimensions with minimized oscillation at the initial stage, as shown in \lcol{Figure}  \ref{pipeline}.  In a forthcoming paper, we plan to delve deeper into the predictions for rigid docking sites in three or more dimensions.
\begin{figure*}[htb]
  \centering
  \begin{tikzpicture}
  \begin{scope}[scale = 0.93, every node/.style={scale=0.93}]]
  \node at (0,0)  {\includegraphics[ width=2.5in, height=2.25in,trim=1.9cm 4.50cm 0.75cm 5cm, clip]{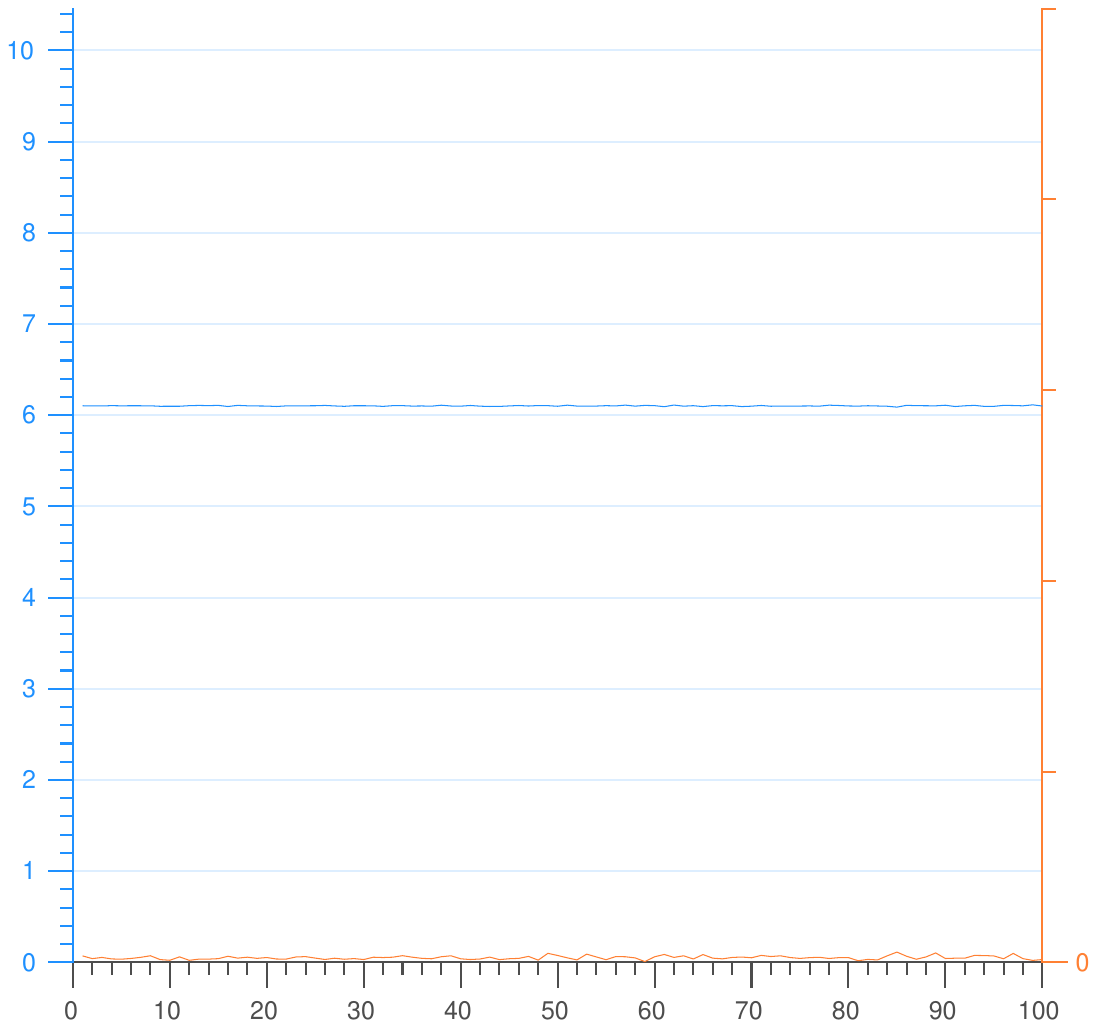}
                   \includegraphics[width=2.5in,height=2.25in,trim=1.9cm 4.50cm 0.75cm 5cm, clip]{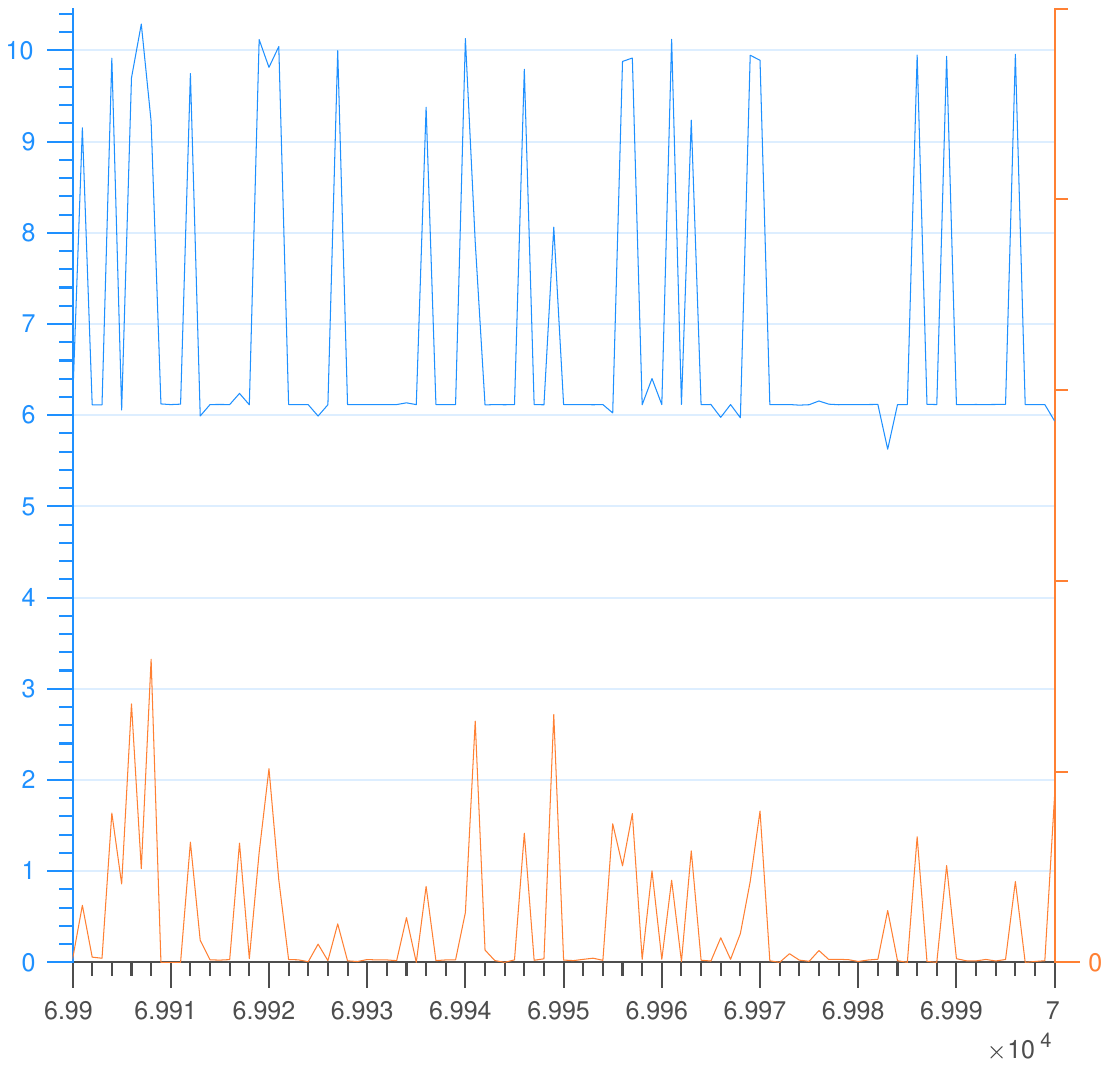}

  };

  \node at (9.7,0) {
                   \includegraphics[width=2.5in,height=2.25in,trim=1.9cm 4.50cm 0.75cm 5cm, clip]{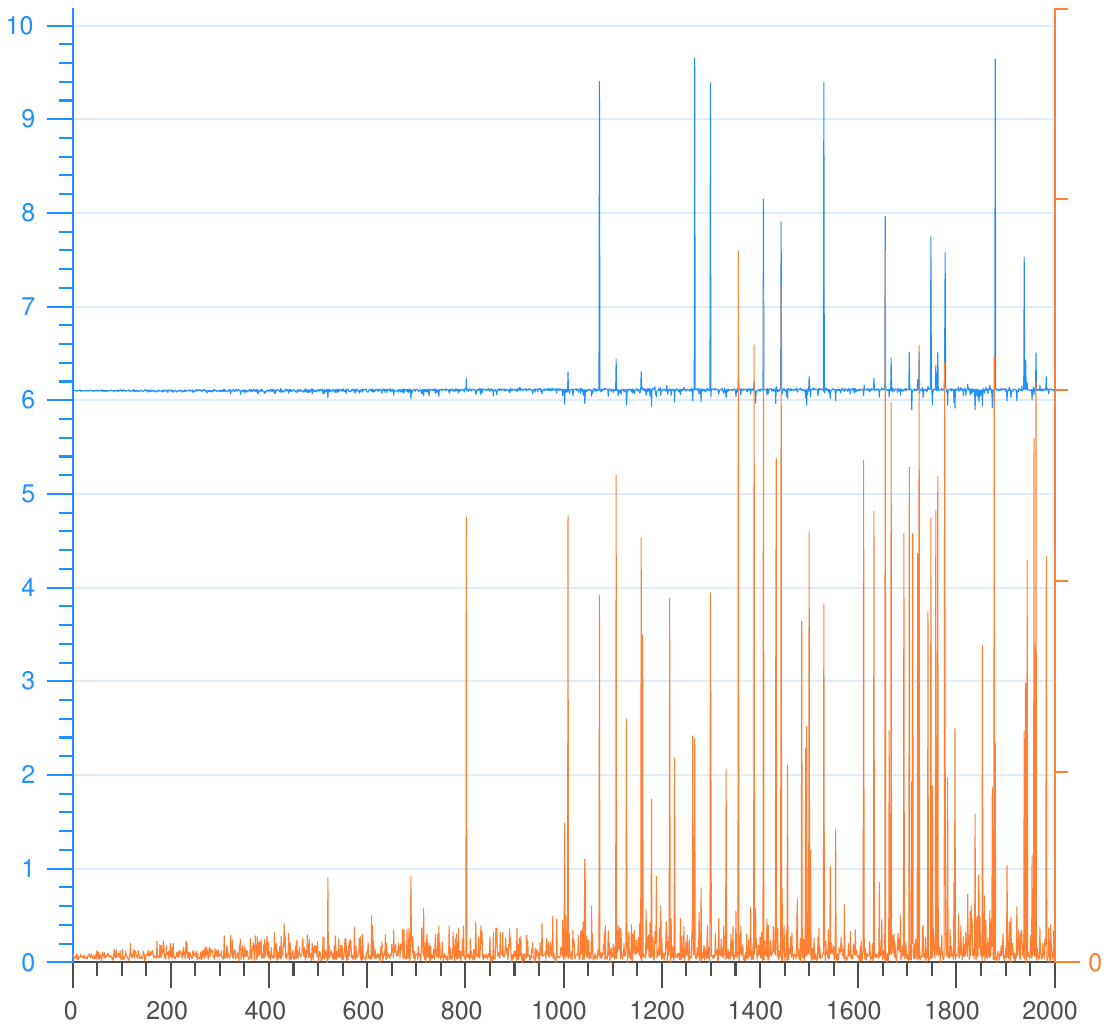}
                   };

  \node at (-5.65,2.75)  {(a)};
  \node at (0.75,2.75)  {(b)};
   \node at (-3,-3.5)  {(c)};
  
  \node at (-5,0.75) [right]  {$\eset{\log (\E(\balpha,\blambda,\omega)+ \delta)}$};
  \node at (-4.5,-3) [right]  {Rotation Number};
  \node at (2,-3) [right]  {Rotation Number};
  \node at (0.6,-0.5) [right]  {$\mbox{SD}{[\log (\E(\balpha,\blambda,\omega)]}$};
  
  \end{scope}
\end{tikzpicture}
    \caption{Trypsin energy manifold uncertainty measures with respect to 3-dimensional stochastic receptor deformations($N_{\mathcal{A}} = 3$), we multiply the random factor of KL expansion by 0.8 to reduce noise. 
The top line is the mean and the bottom the SD. The corresponding scales axis are given to the left and right.  (a) First 100 top sorted rotations. Same as 2 dimensions, the curve is very flat. (b) last 100 sorted rotations. (c) First 2000 top sorted rotations.
}
    \label{fig:3dstats}
\end{figure*}

\section{Discussion}
\label{Discussion}
As we have discussed in the previous sections, adopting a normality assumption changes the probability measure of the stochastic manifold without affecting the molecular manifold's covariance structure. Additionally, a key advantage of this model is its independence from specific ligand and receptor mappings, owing to its foundation on the molecular manifold's geometry.
The model's predictions for three-dimensional stochastic receptor deformations align with those of rigid body docking site analyses. 
The top one hundred rotations exhibit notably less uncertainty compared to others. Nonetheless, the model occasionally assigns higher rankings to rotations not identified as optimal. The accuracy of this framework is primarily affected by three types of errors: model error, truncation error, and integration error.
The stochastic receptor manifold is modeled as a realization of a Gaussian process. This assumption introduces a potential model error, given that such manifolds seldom exhibit Gaussian distribution properties in practice. Consequently,
computations based on the Gaussian probability measure may be incorrect.
Truncation error emerges from the dimensional reduction of the stochastic manifold to a three-dimensional space ($N_{\mathcal{A}} = 3$), capturing only approximately 32\% of the manifold's total dimensional scope.  To reduce truncation and integration errors, it is necessary to increase the number of dimensions.
However, computational expenses rise sharply for dimensions exceeding three. 
Comprehensive stochastic optimization becomes viable with the application of accurate numerical techniques for high-dimensional models, a topic designated for future research endeavors. 
Due to inaccuracies introduced by higher-dimensional analysis, pinpointing the exact optimal docking site is improbable. Nonetheless, the framework remains valuable for assessing the suitability of rotational and translational configurations as potential docking sites.



\begin{remark}
    An alternative approach of evaluating configuration uncertainty is to estimate the empirical distribution across the stochastic deformations, and then compute the probability $\bbP(\log(\E(\balpha, \blambda, \omega) + \delta)
    \leq C)$, for some parameter $C$. The distribution for $\log(\E(\balpha, \blambda, \omega) + \delta)$ can be computed by using the
    RBF representation as a surrogate model. This can be done by generating large amounts of realizations for 
    $Z_1(\omega), \dots, Z_{N}(\omega)$ and from the RBF surrogate model compute realizations of $\log(\E(\balpha, \blambda, \omega) + \delta))$ with reasonable computational burden. By using a histogram the probability can be estimated. This is left as future work.
\end{remark}

\ifCLASSOPTIONcompsoc
  \section*{Acknowledgments}
\else
  \section*{Acknowledgment}
\fi

  Research reported in this
technical report was supported in part by the National Institute of
General Medical Sciences (NIGMS) of the National Institutes of Health
under award number 1R01GM131409-03 and the National Science
Foundation under Grant No. 1736392.

\appendix
\setcounter{equation}{0}
\renewcommand{\theequation}{A.\arabic{equation}}
\label{Appendix}

We begin by presenting the Mercer's Theorem, which plays a crucial role in validating the subsequent theorems and propositions.

\begin{theorem}[Mercer's Theorem]Suppose $K$ is a symmetric positive definite kernel. Then there is a complete orthonormal basis $\{\phi_i\}_{i \in \mathbb{N}}^{\infty}$ of $L_2(D)$ consisting of eigenfunctions of $T_K :L^{2}(D)  \to L^{2}(D): f \to T_{K} = \int_{D} K(\bs,\cdot) f(\bs)ds$, such that the corresponding sequence of eigenvalues $\{\lambda_i\}_i$ is nonnegative. Then $K$ has the representation:
\begin{equation}\label{mercer}
    K(\bs,\bv) = \sum_{j = 1}^{\infty} \lambda_j \phi_j(\bs)\phi_j(\bv)^T.
\end{equation}

\end{theorem}

\begin{proof}
see page 4-6 in \cite{julio2022stochsatic}
\end{proof}

\subsection*{Proof of Multivariate Karhunen-Lo\`{e}ve Theorem (Theorem 1)}
The primary objective of this proof is to show that the truncated Karhunen-Lo\`{e}ve expansion converges to stochastic vector process $u(\bv, \omega)$ in $L_2$. \par
Let $Z_k(\omega) = \int_{D} u(\bv,\omega)^{T}\phi_k(\bv)\id\bv  $, we claim  that the $Z_k$ are uncorrelated with mean zero and variance $\lambda_k$. First, we have 

\begin{equation*}
    \begin{split}
    \eset{Z_k(\omega)} &=\eset{ \int_{D} u(\bv,\omega)\phi_k(\bv) \id \bv}\\ & =   \int_{D} \eset{u(\bv,\omega)}\phi_k(\bv)\id \bv\\ &= 0.   
    \end{split}
\end{equation*}
Furthermore, $\eset{Z_i(\omega) Z_j(\omega)}$
\[
\begin{split}
      &= \eset{ \int_{D}\int_{D} \phi_i(\bv)^{T} u(\bv,\omega) u(\bv',\omega)^{T}\phi_j(\bv')\id \bv \id \bv')} \\ 
      &=  \int_{D}\int_{D} \phi_i(\bv)^{T} \eset{(u(\bv,\omega) u(\bv',\omega)^{T})}\phi_j({\bv}')\id \bv \id \bv'      \\ &=
  \int_{D}\int_{D} \phi_i(\bv)^{T} R(\bv,\bv')\phi_j(\bv')\id \bv \id \bv'\\  &= \int_{D}\phi_i(\bv)^{T}\lambda_{j}\phi_j(\bv) \id\bv \\ &= \delta_{ij}\lambda_j.
\end{split}
  \]
 The equations above come from the fact that $\phi_k$'s are orthonormal eigenfunctions of the operator $T_{R_u}$. Then, by calculating the mean squared error of truncated Karhunen-Lo\`{e}ve expansion, we obtain an expression below with three terms:

\begin{align}
    &\eset{\ \left\|u(\mathbf{v},\omega)-\sum_{k=1}^{\ell}\left(u(\mathbf{v},\omega),\phi_k \right)_{L^2(D)} \phi_{k}\right\|_{L^2(D)}^{2} } \\
    &= \eset{ \left(u(\mathbf{v},\omega),u(\mathbf{v},\omega)\right)_{L^2(D)} }\label{subeq1} \\
    &\quad -2\eset{\left(u(\mathbf{v},\omega),\sum_{k=1}^{\ell}(u(\mathbf{v},\omega), \phi_{k})_{L^2(D)} \phi_{k} \right)_{L^2(D)}} \label{subeq2} \\
    &\quad +\mathbb{E}\Biggl[ \biggl(\sum_{k=1}^{\ell}(u(\mathbf{v},\omega), \phi_{k})_{L^2(D)} \phi_{k}\biggr)_{L^2(D)}, \nonumber \\
    &\quad \quad \biggl. \sum_{k=1}^{\ell}(u(\mathbf{v},\omega), \phi_{k})_{L^2(D)} \phi_{k} \biggr)_{L^2(D)}\Biggr],\label{subeq3} \end{align}

\par


\begin{align*}
  \eqref{subeq1}    &= \eset{ \left(\int_{D}u(\bv,\omega)^T u(\bv,\omega) \id \bv \right)} = Trace(R_u(\bv,\bv)), 
\end{align*}

\begin{align*}
        \eqref{subeq2}
        &= 
        -2\eset{\sum_{k=1}^{\ell}\int_{D}\phi_{k}(\bv)^T u(\bv,\omega)\left(  u(\bv,\omega), \phi_{k}\right)_{L^2(D)}    \id \bv } \\ 
     &= -2\sum_{k=1}^{\ell}\int_{D} \int_{D} \phi_{k}(\bv)^T
      R(\bv,\bv')\phi_{k}(\bv') \id \bv'  \id \bv\\
      &= -2\sum_{k=1}^{\ell}\int_{D} \phi_{k}(\bv)^T 
        \lambda_{k} \phi_{k}(\bv) \id \bv\\ & = -2\sum_{k=1}^{\ell} \lambda_k, \\
\end{align*}

        \begin{align*}
   \eqref{subeq3} &=
  \sum_{k=1}^{\ell}\int_{D}\int_{D}\phi_{k}(\bv)^T \eset{ u(\bv,\omega)u(\bv',\omega)^T  } \phi_{k}(\bv') \id \bv'\id \bv \\
   &= \sum_{k=1}^{\ell}\int_{D} 
        \lambda_{k}\phi_{k}(\bv)^T \phi_{k}(\bv) \id \bv \\
   &=  \sum_{k=1}^{\ell} \lambda_k.
\end{align*}

Then,
\begin{flalign}
    & \eset{ \left\|u(\mathbf{v})-\sum_{k=1}^{\ell}\left((u(\mathbf{v}), \phi_{k}\right)_{L^2(D)} \phi_{k}\right\|_{L^2(D)}^{2}} && \nonumber \\
    &= \eqref{subeq1} - \eqref{subeq2} + \eqref{subeq3} && \nonumber \\
    &= \text{Trace}\left(R_u(\mathbf{v},\mathbf{v})\right) - \sum_{k=1}^{\ell} \lambda_{k} && \nonumber \\
    &= \text{Trace}\left(R_u(\mathbf{v},\mathbf{v})\right) - \text{Trace}\left(\sum_{k=1}^{\ell} \lambda_k \phi_k \phi_k^T\right). &
    \label{eq:kl_final}
\end{flalign}

       The final term in equation \eqref{eq:kl_final} goes to zero as $\ell \rightarrow \infty$ by Mercer's theorem.
        
\subsection*{Proof of Proposition 1}

Suppose we have a centered Gaussian random field $u(\bv)= (u_1(\bv),u_2(\bv),...,u_n(\bv))^{T} \in L^2(\Omega)$, where $\Omega = [a_1,b_1] \times [a_2,b_2] \times ,...,\times [a_n,b_n] \in \mathbb{R}^{n}$,  each $[a_i,b_i]$ is a closed interval in $\mathbb{R}$ for all $1 \leq i \leq n$.
Let  $(.,.)_{L^2(\Omega)}$ be the inner product defined by $(f,g)_{ L^2(\Omega)}= \int_{\Omega} f^{T}g \id x$. \par

Define $$Z_{k}(\bv)= \int_{\Omega} u(\bv)^T\phi_{k}(\bv)dt=\int_{\Omega}\sum_{i=1}^{n} u_{i}(\bv) \phi_{k,i}(\bv) \id \bv. $$
Partition each interval $[a_i, b_i]$ into a finite family $I_i$ of $m_i$ non-overlapping closed subintervals. Consider the finite family of subrectangles, denoted by $C$, which is defined as the Cartesian product of intervals $C = I_1 \times I_2 \times \cdots \times I_n$. Here, we obtain a total of $m = \prod_{i=1}^{n} m_i$ smaller rectangles, each referred to as $C_i$. Consequently, the entire domain $\Omega$ can be expressed as the union of all these smaller rectangles, $\Omega = \bigcup_{i=1}^{m} C_i$.
Then $$Z_k(\bv) = \lim_{|\triangle| \to 0} \sum_{j = 1}^{m} \sum_{ i = 1}^{n} u_i(\bv_j)\phi_{k,i}(\bv_j)m(C_j)$$ where $\bv_j \in C_j$ for all $1\leq j \leq m$, $m(C_j)$ is the measure (area) of $C_j$, $\triangle$ is the largest partition cube, and $|\triangle|$ represents the area of that cube.

Define \(d_{i,j} = \phi_{k,i}(\mathbf{v}_j)m(C_j)\), and let 
\[
R^{k}(\Delta) = \sum_{j = 1}^{m} \sum_{i = 1}^{n} u_i(\mathbf{v}_j)\phi_{k,i}(\mathbf{v}_j)m(C_j) = \mathbf{D} \cdot \mathbf{U}^T,
\]
where
\[
\mathbf{D} = (d_{1,1}, d_{1,2}, \ldots, d_{1,m}, d_{2,1}, \ldots, d_{2,m}, \ldots, d_{n,1}, \ldots, d_{n,m})
\]
and the vector \(\mathbf{U}\) is constructed as follows:
\[
\begin{aligned}
\mathbf{U} = (&u_1(\mathbf{v}_1), u_1(\mathbf{v}_2), \ldots, u_1(\mathbf{v}_m), \\
&u_2(\mathbf{v}_1), u_2(\mathbf{v}_2), \ldots, u_2(\mathbf{v}_m), \\
&\ldots, u_n(\mathbf{v}_1), u_n(\mathbf{v}_2), \ldots, u_n(\mathbf{v}_m)).
\end{aligned}
\]

    Since $u(\bv)$ is a Gaussian random field,  $u_i(\bv_j)$ is also Gaussian. And since linear transformation preserves the Gaussian property of random variables, $R^{k}(\triangle)$ is also a zero mean Gaussian random variable. To compute its variance, we have

    \begin{equation*}
        \begin{split}
          \text{Var}(R^k(\triangle)) &= \eset{(R^k(\triangle))^2}\\ &= \eset{\sum_{j = 1}^{m}\phi^{k}(\bv_j)^T u u^{T} \phi^k(\bv_j)) m(C_j)^2}\\ &= \sum_{j = 1}^{m}\phi^{k}(\bv_j)^T\eset{ u u^{T}} \phi^k(\bv_j) m(C_j)^2\\  &= \sum_{j = 1}^{m}\phi^{k}(\bv_j)^T K(\bv_j,\bv_j) \phi^k(\bv_j) m(C_j)^2  
        \end{split}
    \end{equation*}
    and
\begin{equation*}
\begin{split}
        \lim_{|\triangle| \to 0 }\text{Var}(R^{k}(\triangle)) &= \int_{\Omega}\phi^{k}(\bv)^T K(\bv,\bv) \phi^k(\bv)\id \bv \\ &= \phi^k(\bv)^T\lambda_k \phi^k(\bv)\\ &= \lambda_k.
\end{split}
\end{equation*}

By computing the characteristic function of $R^k(\triangle)$ and taking a limit as $\left|\triangle \right| \to 0$, we get:
\begin{equation*}
    \begin{split}
       \lim_{|\triangle| \to 0 }\eset{e^{itR_{k}(\triangle)}} &= \lim_{|\triangle| \to 0 } e^{-\frac{1}{2} \text{Var}(R^k(\triangle))t^{2}}\\ &= e^{-\frac{1}{2}\lambda_k t^2}. 
    \end{split}
\end{equation*}
    Thus $Z_k \sim \mathcal{N}(0,\lambda_k)$.
\subsection*{Proof of Proposition 2}
\hfill \\
See Theorem 2.7 in \cite{Schwab2006}
\subsection*{Method of Snapshots}\label{snapshots}
To apply the Karhunen-Lo\`{e}ve expansion to the stochastic molecular domain, we need to solve the equation \eqref{Fredholm} to get eigenfunctions $\phi_{k}(\bv)$.
The Method of Snapshots, developed by by Sirovich, reduces this infinite-dimensional problem of computing the eigenfunctions $\phi_{k}(\bv)$ to an eigenvalue decomposition problem of a finite matrix by taking samples $u(\bv,t_k)$ at different discrete times $t_k$.
Assume that $u(\bv)$ is an ergodic process. that is, the time average is equal to the ensemble average, the correlation can be reduced to:

\begin{align*}R_u\left(\bs, \bv\right)=\lim_{M \to \infty}\frac{1}{M} \sum_{i=1}^{M} u^{i}(\bv) u^{i}\left(\bs\right)^{T}\end{align*}
where $u^{i}(\bv) = (u^{i}_{1}(\bv), u^{i}_{2}(\bv), u^{i}_{3}(\bv), ..., u^{i}_{l}(\bv))^{T}$ for $i = 1,2,...,M$.
With finite number of samples (snapshots),  $R_u(\bs,\bv)$ can be approximated as the following:
\begin{align}\label{finitesum}
    R_u\left(\bs, \bv\right)\approx \frac{1}{M} \sum_{i=1}^{M} u^{i}(\bv) u^{i}\left(\bs\right)^{T}
\end{align}

By substituting the equation \eqref{finitesum}, for each $k$, the equation \eqref{Fredholm}  becomes:
\begin{align*}
  \frac{1}{M} \sum_{i=1}^{M} u^{i}(\bv)\int_{D} u^{i}\left(\bs\right) ^{T} \tilde{\phi}_k\left(\bs\right) \id \bs=\lambda_k \tilde{\phi}_k(\bv).
\end{align*}\\
where $\tilde{\phi}_k$ is the corresponding eigenfunction to the approximation of covariance matrix function.

Define $a_{i,k} =\int_{D} u^{i}\left(\bs\right)^{T} \tilde{\phi}_k \left(\bs\right) d \bs $, then,
\begin{align}\label{simplified}
\frac{1}{M}  \sum_{i=1}^{M} a_{i,k}u^{i}(\bv) = \lambda_k \tilde{\phi}_k(\bv).\end{align}
 Multiplying both sides by $u^{j}(\bv)^{T}$ and integrating over $D$ gives:

 \[
 \begin{split}
     \frac{1}{M} \sum_{i=1}^{M} a_{i,k}\int_{D}u^{j}(\bv)^{T}u^{i}(\bv) \id\bv &= \lambda_k \int_{D} u^{j}(\bv)^{T}\tilde{\phi}_k(\bv) \id \bv. \\ \end{split}
 \]
 
     Define 
     \[
     \begin{split}
C_{ji}&=\frac{1}{M} \int_{D} u^{j}(\bv)^{T} u^{i}(\bv) \id \bv \\  &= \sum_{r = 1}^{l} \frac{1}{M} \int_{D} u^{j}_{r}(\bv)^{T} u^{i}_{r}(\bv) \id \bv \\
     \end{split}
     \]
      for $ i, j=1, \ldots, M$.\\ Then the equation \eqref{simplified} becomes:

$$
         \sum_{i=1}^{M} C_{j i} a_{i,k} = \lambda_k a_{j,k},
    $$
    that is
  \begin{equation}\label{eq:pod2}
  \left[\begin{array}{ccc}
C_{11} & \dots & C_{1 M} \\
\vdots & & \vdots \\
C_{M 1} & \cdots & C_{M M}
\end{array}\right]\left[\begin{array}{c}
a_{1,k} \\
\vdots \\
a_{M,k}
\end{array}\right]=\lambda_k\left[\begin{array}{c}
a_{1,k} \\
\vdots \\
a_{M,k}
\end{array}\right].\end{equation}
By now, we have transferred the problem in the equation \eqref{Fredholm} to the problem of finding eigenvalues and eigenvectors of an $M \times M$ matrix showed in the equation \eqref{eq:pod2}. Most of the time the latter is much easier to compute than the former since the number of samples $M$ is usually much smaller than the length of $u(\bv)$. And as observed, the approximation of orthonormal eigenfunctions are of the form:  
  \begin{align*}
 \tilde{\phi}_k(\bv)=\frac{1}{\lambda_k M} \sum_{i=1}^{M} a_{i,k}u^{i}(\bv).
  \end{align*}

\bibliographystyle{IEEEtran}
\bibliography{citations}
\begin{IEEEbiography}
   [{\includegraphics[width=1in,height=1.25in,clip,keepaspectratio]{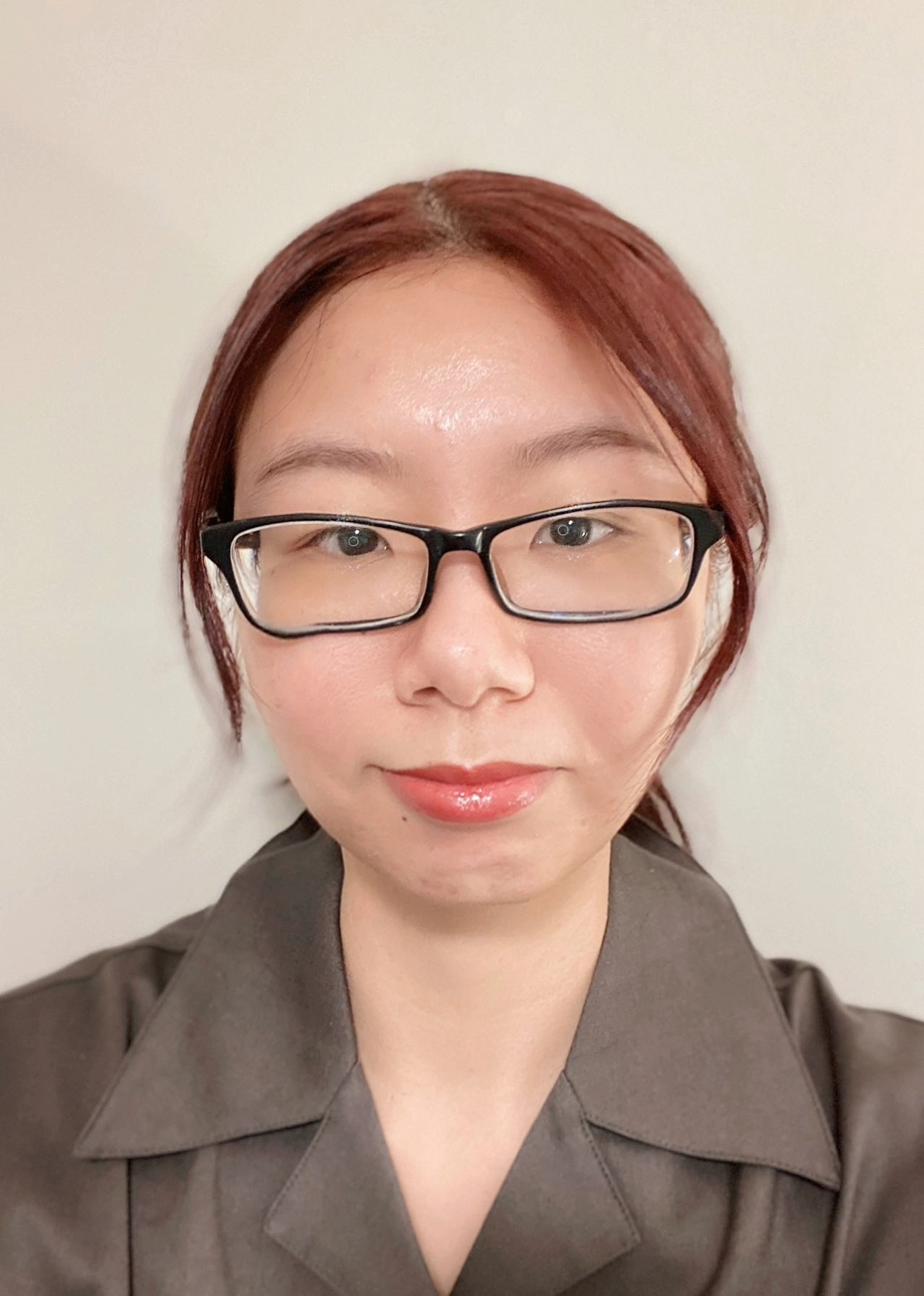}}] {Nanjie Chen} received her Bachelor's degree in Mathematics at Indiana University Bloomington in 2018. She is currently
working toward the Ph.D. degree in Applied Mathematics at Boston University from September
2020. 
\end{IEEEbiography}
\begin{IEEEbiography}
[{\includegraphics[width=1in,height=1.25in,clip,keepaspectratio]{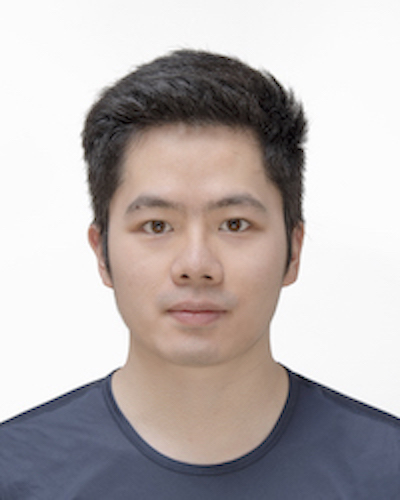}}]{Dongliang Yu} Dongliang Yu was born on April 1th, 1995, in China. He received his B.Sc. in Information and Computing Science from Zhongnan University of Economics and Law, China, in 2017. Pursuing further education, he completed his M.Sc. in Statistics at Department of Mathematics and Statistics, Boston University in 2020 and then went to Department of Mathematics, Stony Brook University to continue his study in field of Math. Dongliang’s research is primarily centered on numerical integral methods in high dimensions.
\end{IEEEbiography}
\begin{IEEEbiography}
[{\includegraphics[width=1in,height=1.25in,clip,keepaspectratio]{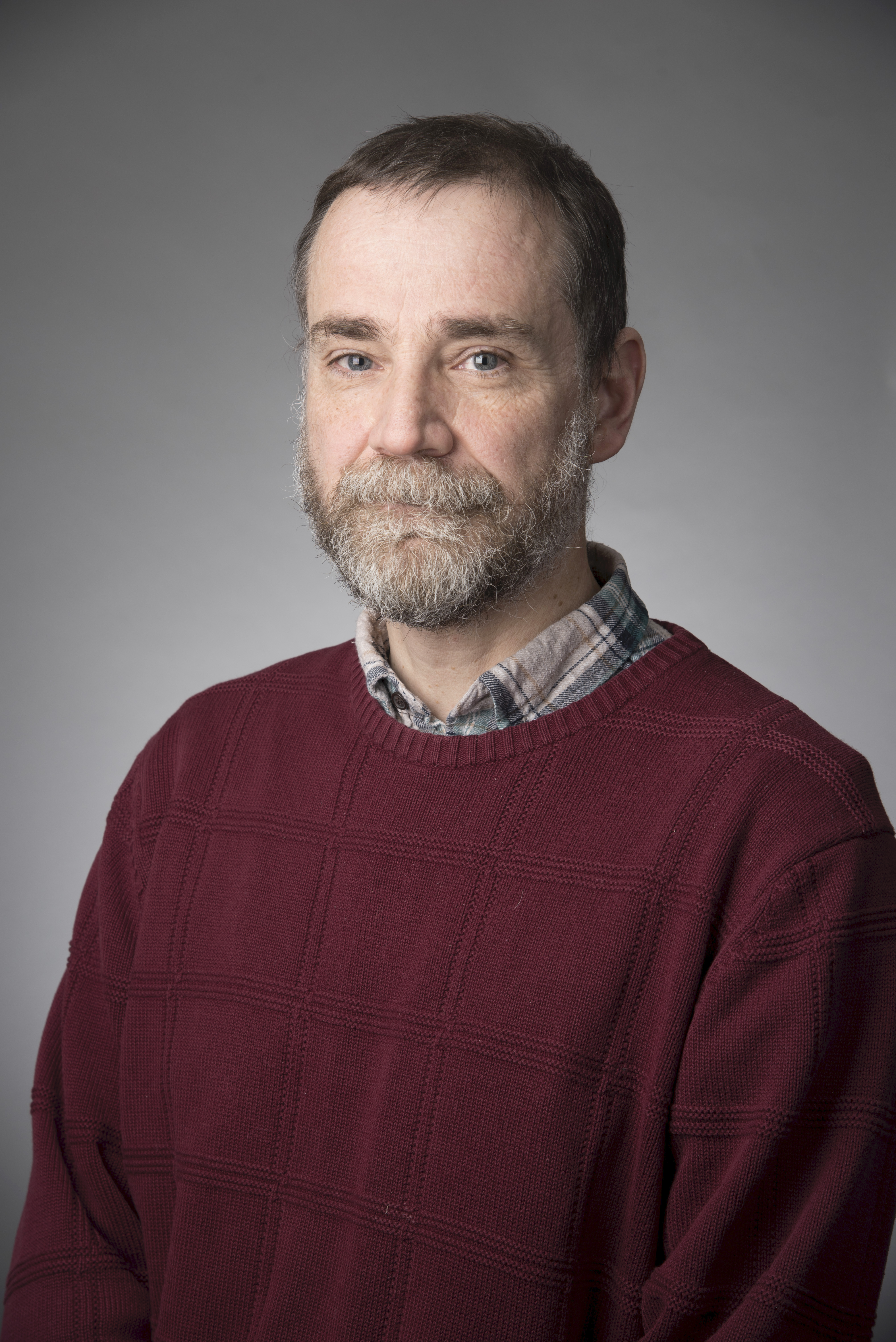}}]{Dmitri Beglov}received M.S. in Chemistry (1987) and Ph.D. in Biophysics (1991) from Moscow Institute of Physics and Technology, Moscow, Russia. From 1992 to 1998 he was a postdoc in Computational Chemistry in the Lab of Benoit Roux at the University of Montreal, Canada. From 1998 to 2006 he was working as a computational chemist holding industrial positions in several companies. He designed, developed and implemented various computational methods in macromolecular modeling programs such as CHARMM. From 2006 he is affiliated with the Lab of Sandor Vajda at Boston University, currently holding a position of Research Assistant Professor.  He is a developer and experienced user of protein docking programs Piper/ClusPro and computational solvent mapping program FTMap. In addition, since 2019 he is a CEO of Acpharis Inc., a US company, developing commercial software products and offering services for pharmaceutical and biotechnology clients. 
\end{IEEEbiography}
\begin{IEEEbiography}[{\includegraphics[width=1in,height=1.25in,clip,keepaspectratio]{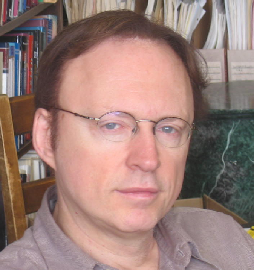}}]{Mark Kon} obtained Bachelor’s degrees in Mathematics,
Physics, and Psychology from Cornell
University, a PhD in Mathematics from MIT. He
is a professor of Mathematics and Statistics
at Boston University. He is affiliated with the
Quantum Information Group, the Bioinformatics
Program and the Computational Neuroscience
Program. He has had appointments at Columbia
University as Assistant and Associate Professor
(Computer Science, Mathematics), as well as at
Harvard and at MIT. He has published approximately
100 articles in mathematical physics, mathematics and statistics, 
computational biology, and computational neuroscience, including two
books. His recent research and applications interests involve quantum
probability and information, statistics, machine learning, computational
biology, computational neuroscience, and complexity.
\end{IEEEbiography} 
\begin{IEEEbiography}[{\includegraphics[width=1in,height=1.25in,clip,keepaspectratio]{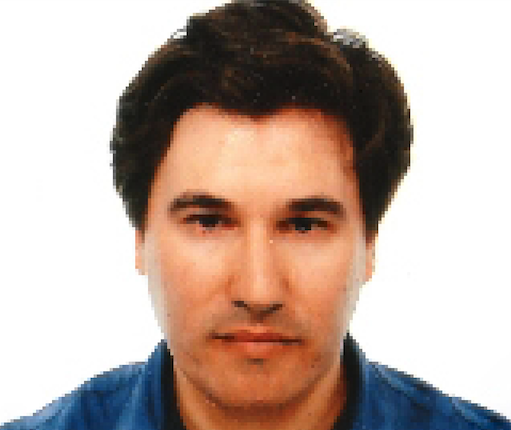}}]{Julio Enrique Castrill\'on-Cand\'as} received the MS and Ph.D. degrees in electrical engineering and computer science from the Massachusetts Institute of Technology (MIT), Cambridge. He is currently faculty in the department of Mathematics and Statistics at Boston University. His area of expertise is in Uncertainty Quantification (PDEs, non-linear stochastic networks), large scale computational statistics, functional data analysis and statistical machine learning.
 
\end{IEEEbiography}

\end{document}